\documentclass[12pt]{article}

\usepackage{amsmath,amsthm,amssymb,mathrsfs,bbm}
\usepackage[toc, page]{appendix}
\usepackage{newpxtext,newpxmath}
\usepackage{geometry,pdflscape,setspace}
\usepackage{graphicx,subfig}
\usepackage{fancyhdr}
\usepackage{authblk}
\usepackage[toc, page]{appendix}
\usepackage[normalem]{ulem}
\usepackage[format=hang, textfont={md}, labelfont={bf}, justification=raggedright]{caption}
\usepackage{natbib,threeparttable,dcolumn,multicol,lscape,bigstrut,booktabs,multirow,colortbl,array}
\usepackage[colorlinks=true,citecolor=blue,linkcolor=blue,%
    anchorcolor=black,urlcolor=blue,pdfborder={0 0 0},pdfstartview=FitH]{hyperref}

\graphicspath{{./Figure/}}

\theoremstyle{plain}
\newtheorem{theorem}{Theorem}%[section]
\newtheorem{lemma}{Lemma}%[section]
\newtheorem{proposition}{Proposition}%[section]
\newtheorem{condition}{Condition}
\theoremstyle{definition}
\newtheorem{assumption}{Assumption}%[section]
\theoremstyle{remark}
\newtheorem{remark}{Remark}%[section]
\newtheorem{example}{Example}%[section]

\newcommand{\indep}{\perp \!\!\! \perp}

\begin{document}
\title{Testing Identifying Assumptions in Parametric Separable Models: A Conditional Moment Inequality Approach}
\author{Leonard Goff \thanks{University of Calgary, Canada. Email address: \href{mailto: leonard.goff@ucalgary.ca}{leonard.goff@ucalgary.ca}},\ \ D\'{e}sir\'{e} K\'{e}dagni \thanks{University of North Carolina at Chapel Hill, the United States. Email address: \href{mailto: dkedagni@unc.edu}{dkedagni@unc.edu}},\ \ and
Huan Wu \thanks{University of North Carolina at Chapel Hill, the United States. Email address: \href{mailto: huan.wu@unc.edu}{huan.wu@unc.edu}}}
%\affil{University of North Carolina at Chapel Hill}
\date{First version:\thanks{We are grateful to Andrii Babii, Eric Mbakop, Adam Rosen, and Valentin Verdier for their helpful comments. We are also grateful to participants at the UNC-Chapel Hill Econometrics Workshop. All errors are ours.} \ March 9, 2023\\ \vspace{0.25cm} This version: \today}

\maketitle
\onehalfspacing
% Can also choose \singlespacing or \doublespacing

%%%%%%%%%%%%%%%%%%%%%%%%%%%%%%%%%%%%%%%%%%%%%%%%%%%%%%%%%%%%%%%%%%%%%%%%
%%%% Abstract

\begin{abstract}
    In this paper, we propose a simple method for testing identifying assumptions in parametric separable models, namely treatment exogeneity, instrument validity, and/or homoskedasticity. %We derive identification results after relaxing some of the identifying assumptions when they are rejected by the data. 
We show that the testable implications can be written in the intersection bounds framework, which is easy to implement using the inference method proposed in \cite{chernozhukov2013intersection}, and the Stata package of \cite{chernozhukov2015implementing}. Monte Carlo simulations confirm that our test is consistent and controls size. We use our proposed method to test the validity of some commonly used instrumental variables, such as the average price in other markets in \cite{nevo2012identification}, the Bartik instrument in \cite{card2009immigration}, and the test rejects both instrumental variable models. When the identifying assumptions are rejected, we discuss solutions that allow researchers to identify some causal parameters of interest after relaxing functional form assumptions. We show that the IV model is nontestable if no functional form assumption is made on the outcome equation, when there exists a one-to-one mapping between the continuous treatment variable, the instrument, and the first-stage unobserved heterogeneity.
\end{abstract}

 \maketitle
{\footnotesize \textbf{Keywords}: Separable models, functional form, OLS, IV, testable implications, relaxed assumptions.

}

\clearpage
%%%% Introduction

\section{Introduction}\label{sec: intro}
Instrumental variable (IV) models are widely used in economics and related fields. In applied work, researchers often specify functional forms for the relationship between the outcome variable and the regressor of interest. The instrumental variable assumptions (exogeneity, exclusion restriction, and relevance) combined with functional form restrictions impose testable implications on the observed joint distribution of the outcome, the regressor, and the instrument.

Testing identifying assumptions is familiar in applied work. For example in the linear IV model, researchers usually test the relevance condition, which states that the covariance between the regressor and the instrument differs from zero. Accordingly, many papers in the literature propose solutions for how to deal with weak IVs (see for example \cite{staiger1997instrumental, wang1998inference, stock2002testing, hansen2014instrumental, andrews2019weak}, etc.). Although there has been some attention devoted to testing the exogeneity and exclusion restriction assumptions, applied researchers do not often report a test result for these assumptions. A possible explanation could be that the existing approaches are difficult to implement in practice.

In this paper, we propose an easy-to-implement testing procedure for the IV assumptions in a broad class of parametric separable models. To do so, we transform the testable conditional moment equality implication of the model into two conditional moment inequalities, as this is more conducive to inference given recent developments in the moment inequality literature \citep{AS2013, chernozhukov2013intersection}. The test can be implemented using existing Stata packages developed by \cite{chernozhukov2015implementing}. Our approach relies on a simple two-step procedure. We first identify the model parameters using the standard IV methods, and we then plug the identified coefficients in the outcome equation to back out the error term. Afterwards, we check the exogeneity condition by testing nonparametrically the joint conditions that i) the supremum of the conditional expectation of the error term given the instrument values is less than or equal to zero; and ii) its infimum is bigger than or equal to zero.

The test is asymptotically consistent as the IV estimator converges at a parametric rate, and the approximation error that comes from replacing the true coefficients with the estimated ones vanishes when the sample size goes to infinity. We show through simulations how the test controls the size asymptotically (sample size $\gtrapprox$ 3000) and the power converges to one for large samples (1000 or bigger). We illustrate the proposed test on two real-world empirical examples: the average price in other markets IV \citep{nevo2012identification}, and the Bartik IV \citep{card2009immigration}. %\footnote{The sample sizes are respectively 990 in \cite{nevo2012identification} and 124 in \cite{card2009immigration}.} 
The test rejects the validity of both IVs. Finally, we discuss alternative solutions that are available to researchers when a parametric IV model is rejected. A rejection of the IV model could be due to a violation of the IV assumptions or functional form misspecification. \cite{nevo2012identification}, \cite{Conley2012}, \cite{Masten2021}, among others have proposed alternative identification results when the IV is invalid. In this paper, we propose nonparametric identification when the IV is valid, but functional forms may be misspecified. Our approach is similar to \cite{imbens2009identification}. We relax all functional form assumptions for the outcome equation but assume that the relationship between the regressor and first-stage unobserved heterogeneity is continuous and strictly monotonic for all values of the instrument. We then point-identify the marginal response function for each value of the regressor and thus various marginal treatment effects.  

The testability of the instrumental variable model assumptions has been questioned over the years. Some researchers think the exogeneity and exclusion restriction assumptions are not testable because they relate the IV to an unobserved variable (potential outcomes or the error term in a linear model). However, a lot of progress has been made in recent years to elucidate when these IV assumptions can be testable or untestable. \cite{pearl1995testability} appears to be the first paper that derives a testable implication for the IV independence (exogeneity) and exclusion restriction assumptions when the endogenous regressor is discrete. \cite{kedagni2020generalized} added a new set of testable implications to those in \cite{pearl1995testability} and showed that their testable implications are sharp when the outcome and treatment are binary but the instrument is unrestricted.  \cite{gunsilius2021nontestability} proved that the instrumental variable independence assumption is not testable when the endogenous treatment is continuously distributed and there are no structural assumptions. However, most of the models that researchers consider in applied work impose some additional structures that help identify the parameters of interest. For example, in the local average treatment effect model, the IV independence assumption is often coupled with the monotonicity assumption, which together yield testable implications as discussed in \cite{kitagawa2015test}, \cite{mourifie2017testing}, and \cite{Huber2015TestingConstraints}. See also \cite{acerenza2023testing} in the context of bivariate probit models, \cite{Arai2022TestingDesigns}, and \cite{Hsu2023TestingDesigns} in the regression discontinuity design framework. All the above mentioned papers focus on nonseparable models. In this paper, we focus on testing the IV assumptions in parametric separable models.  

In the existing literature on testing identifying assumptions in parametric IV regression models, a conventional approach is to transform conditional moment restrictions to unconditional ones using weighting functions of the instrumental variable and then testing unconditional moment conditions (see the discussions in \cite{bierens1982consistent, bierens1990consistent}, \cite{newey1985maximum}, and related literature). Researchers should be cautious about choosing instrumental functions with this method since the tests could be inconsistent with most alternatives if improper instrumental functions are selected. Another standard method is to estimate the conditional moments with smoothing techniques, such as the kernel smoothed method (for example, \cite{hardle1993comparing}, \cite{horowitz2001adaptive} and related literature) and the smoothed empirical likelihood ratio method (see \cite{tripathi2003testing} and \cite{kitamura2004empirical}). 

Despite available methods, most applied papers do not implement them to test identifying assumptions in regression models. When using instrumental variable methods, most applied papers focus on justifying the validity of their instrument based on some intuition. This motivates us to propose an easy-to-implement method that allows researchers to check whether or not their model assumptions are compatible with the data, and propose a relaxed IV model that still offers identification if a test rejection is thought to arise from functional form misspecification rather than a failure of instrument validity. 

The remainder of the paper is organized as follows. Section \ref{sec: analytical_framework} presents the testable implications and the testing procedure in linear and nonlinear separable IV models. In Section \ref{sec: MonteCarlo}, we illustrate the power and the size of the proposed test through some Monte Carlo simulations. Section \ref{sec:id_relax_assumption} discusses results that relax the functional form assumption for the outcome equation. Section \ref{sec: empirical} presents the empirical illustrations and Section \ref{sec: futurework} concludes.

\section{Testable Implications in Parametric Separable Models}\label{sec: analytical_framework}

In this section, we describe our method of testing the identifying assumptions, including exogeneity, instrument validity, and homoskedasticity assumptions, in some commonly used structural models. Our method is mainly based on the intersection bounds framework. We first identify the model parameters under the imposed identifying assumptions and then write the unobserved error term as a known function of observed variables and identified parameters. We derive the testable implications in a set of conditional moment inequalities, so that we can use the conditional moment inequality method developed in \cite{chernozhukov2013intersection}.

\subsection{Simple Linear IV Model}
To give the general idea of our test, we begin with a simple example. Consider the basic linear model 
\begin{equation} \label{eq:linear_iv}
    Y = \beta_0 + \beta_1 X +U,
\end{equation}
where $Y \in \mathcal{Y}$ is an observed outcome variable with $\mathbb{E}|Y| < \infty$, $X \in \mathcal{X}$ is a scalar observed potentially endogeneous covariate, and $U$ is an unobserved error term. Equation \eqref{eq:linear_iv} could, for example, represent a treatment effect model in which the effect $\beta_1$ of increasing $X$ by one unit is common to all individuals and constant across levels of $X$.

Suppose that researchers propose an instrument variable $Z \in \mathcal{Z}$ to identify the parameters $(\beta_0, \beta_1)$, with $Z$ satisfying the following assumptions:

\begin{assumption}[Exogeneity] \label{ass:mean_indep}
    $\mathbb{E}[U \mid Z]=0$ almost surely.
\end{assumption}

\begin{assumption}[Relevance] \label{ass:relevance}
    $\operatorname{Cov}(X, Z) \neq 0$.
\end{assumption}
\noindent Note that a third assumption implicit in interpreting Equation \eqref{eq:linear_iv} causally is the so-called \textit{exclusion restriction} that $Z$ does not have a direct effect on $Y$, only affecting it through $X$.\\ 
 
Under the model structure Equation (\ref{eq:linear_iv}) and Assumptions \ref{ass:mean_indep} - \ref{ass:relevance}, one can point identify $(\beta_0,\beta_1)$ as  
\begin{equation*}
    \beta_1=\frac{\operatorname{Cov}(Y, Z)}{\operatorname{Cov}(X, Z)}, \quad \beta_0=\mathbb{E}[Y]-\beta_1 \mathbb{E}[X].
\end{equation*}
Importantly, we can then express $U$ as a function of the data by substituting in the identified parameters, i.e., $U=Y - \mathbb{E}[Y] - (X - \mathbb{E}[X]) \operatorname{Cov}(Y, Z) / \operatorname{Cov}(X, Z)$. Combining this with Assumption \ref{ass:mean_indep}, we obtain the testable implication that 
\begin{equation} \label{eq:implication}
        \mathbb{E}\left[\left.(Y - \mathbb{E}[Y]) - \frac{\operatorname{Cov}(Y, Z)}{\operatorname{Cov}(X, Z)}(X - \mathbb{E}[X]) \right| Z = z\right] = 0,
\end{equation}
for all $z \in \mathcal{Z}$.

The question we may ask at this point is whether the implication \eqref{eq:implication} can be rejected by the data. If $Z$ is binary, then implication \eqref{eq:implication} holds automatically and cannot be rejected.\footnote{This can be shown directly by noting that when $Z$ is binary $\frac{\operatorname{Cov}(Y, Z)}{\operatorname{Cov}(X, Z)} = \frac{\mathbb{E}[Y|Z=1]-\mathbb{E}[Y|Z=0]}{\mathbb{E}[X|Z=1]-\mathbb{E}[X|Z=0]}$, and applying the law of iterated expectations over $Z$ for $\mathbb{E}[Y]$ and $\mathbb{E}[X]$.} 
%$$\mathbb{E}[Y|Z=z] - \mathbb{E}[Y] - (\mathbb{E}[Y|Z=1]-\mathbb{E}[Y|Z=0])\frac{\mathbb{E}[X|Z=z] - \mathbb{E}[X]}{\mathbb{E}[X|Z=1]-\mathbb{E}[X|Z=0]} = (\mathbb{E}[Y|Z=1]-\mathbb{E}[Y|Z=0])\cdot(P(Z=1-z)-P(Z=1-z))=0}

By contrast, if there are at least three elements in the support of $Z$, the condition \eqref{eq:implication} may fail to hold. To see this, suppose that the support of $Z$ is equal to $\{0,1,2\}$. We can use two conditions (e.g., $\mathbb E[U\vert Z=0]=0$ and $\mathbb E[U\vert Z=1]=0$) to identify $\beta_0$ and $\beta_1$, and use a third condition (e.g., $\mathbb E[U\vert Z=2]=0$) to check whether the model is rejected or not. Intuitively, testability arises from the linearity assumption inherent in \eqref{eq:linear_iv}, which only provides over-identification given information provided by more than two moments.

The above discussion leads to the following proposition.

\begin{proposition}\label{prop1}
Consider the model specification \eqref{eq:linear_iv} together with Assumptions \ref{ass:mean_indep} - \ref{ass:relevance}. Then, the testable implication \eqref{eq:implication} holds and is sharp. On the other hand, whenever condition \eqref{eq:implication} holds, there exist a random vector $(X^*,Y^*,U^*,Z)$ and a parameter vector $(\beta_0^*,\beta^*_1)$ s.t. $Y^*=\beta^*_0+\beta^*_1 X^* + U^*$, $(X^*,Y^*,Z)$ has the same distribution as $(X,Y,Z)$, and $E[U^*\vert Z]\neq 0$.
\end{proposition}
\begin{proof}
    See Appendix \ref{app:prop1}.
\end{proof}

\noindent \begin{remark} Proposition \ref{prop1} shows that the standard linear IV model can generally be tested using the implication in Equation \eqref{eq:implication}. However, this is not the same as saying that the model is \textit{confirmed} when the implication holds.
\end{remark}

\noindent As more assumptions are added to the model, it may be possible to extend its testable impications. For example, one could add more structure to the above IV model by assuming homoskedasticity. 
\begin{assumption}[Homoskedasticity] \label{ass:ivhomo}
    $\mathbb{E}\left[U^2 \mid Z\right]=\sigma_0^2 < \infty$ almost surely.
\end{assumption} 
\noindent Under Assumption \ref{ass:ivhomo}, we identify $\sigma^2_0$ as $\mathbb{E}\left[\left((Y - \mathbb{E}[Y]) - (X - \mathbb{E}[X]) \operatorname{Cov}(Y, Z) / \operatorname{Cov}(X, Z) \right)^2\right]$. Then, a testable implication for model Equation (\ref{eq:linear_iv}) and Assumptions \ref{ass:mean_indep} - \ref{ass:ivhomo} is
\begin{equation}
    \begin{aligned} \label{eq:implication2}
        & \mathbb{E}\left[(Y - \mathbb{E}[Y]) - \frac{\operatorname{Cov}(Y, Z)}{\operatorname{Cov}(X, Z)}(X - \mathbb{E}[X]) \mid Z = z\right] = 0,\\
        & \mathbb{E}\bigg[\left((Y - \mathbb{E}[Y]) - \frac{\operatorname{Cov}(Y, Z)}{\operatorname{Cov}(X, Z)}(X - \mathbb{E}[X])\right)^2\\
        & \qquad \qquad \qquad -\mathbb{E}\left[\left((Y - \mathbb{E}[Y]) - \frac{\operatorname{Cov}(Y, Z)}{\operatorname{Cov}(X, Z)}(X - \mathbb{E}[X])\right)^2\right] \mid Z =z\bigg] = 0,
 \end{aligned}
\end{equation}
for all $z \in \mathcal{Z}$.\\

\noindent \begin{remark} All the results derived in this section also hold if $Z$ is replaced by $X$, which implies that the simple linear regression model \eqref{eq:linear_iv} combined with the exogeneity $(\mathbb E[U\vert X]=0)$, relevance $(Var(X)>0)$ and/or homoskedasticity $(\mathbb E[U^2\vert X]=\sigma^2_0)$ are also testable.
\end{remark}

\noindent Before extending the basic idea to more general classes of models, we in the next section propose how to implement a practical test of the moment equalities \eqref{eq:implication} and \eqref{eq:implication2} by re-expressing them as a set of moment inequalities.

%We can show that the testable implication in Eq. (\ref{eq:implication}) is sharp. 

%\subsection{General Structural Models}
%Consider a general structural model with an additively separable error term, 
%\begin{equation} \label{eq:general}
%    Y = m(X, \beta_0) + U,
%\end{equation}
%where $(Y, X)$ on $\mathbb{R} \times \mathbb{R}^k$, $k \geq 1$, are observable, and $U$ is an unobserved error term. Function $m(\cdot, \beta)$ on $\mathbb{R}^{k} \times \Omega$ is continuous and known to researchers. Researchers sometimes propose an instrument variable $Z$ to identify the parameter $\beta_0$ under the following assumptions.
%
%\begin{assumption} \label{ass:mean_indep_gen}
%    There exists an observed random variable $Z \in \mathbb{R}^l$ with $l \geq k$ such that $\mathbb{E}[U \mid Z]=0$ and $\operatorname{Cov}(X, Z) \neq 0$.
%\end{assumption}
%
%\begin{assumption} \label{ass:unique_gen}
%    There is only one $\beta_0 \in \Omega$ such that Eq. (\ref{eq:general}) holds.
%\end{assumption}
%
%We can identify the parameter $\beta_0$ in Eq. (\ref{eq:general}) with Assumptions \ref{ass:mean_indep_gen} - \ref{ass:unique_gen}. Then the testable implication is
%\begin{equation} \label{eq:implication_gen}
%        \mathbb E \left[Y-m(X, \beta_0) \mid Z \right] = 0.
%\end{equation}
%We can also prove that the testable implication in Eq. (\ref{eq:implication_gen}) is sharp.

\subsection{Testing Procedure}
We propose an easy-to-implement approach based on conditional moment inequalities using the intersection bound framework, which does not rely on a choice of instrumental functions to transform them into unconditional moments. Indeed, the implication \eqref{eq:implication} is equivalent to the following.
\begin{equation} \label{eq:implication3}
    \begin{aligned} 
      & \sup _{z \in \mathcal{Z}} \mathbb{E}\left[(Y - \mathbb{E}[Y]) - \frac{\operatorname{Cov}(Y, Z)}{\operatorname{Cov}(X, Z)}(X - \mathbb{E}[X]) \mid Z=z\right]  \leq 0, \\
       & \sup _{z \in \mathcal{Z}} \mathbb{E}\left[-(Y - \mathbb{E}[Y]) + \frac{\operatorname{Cov}(Y, Z)}{\operatorname{Cov}(X, Z)}(X - \mathbb{E}[X]) \mid Z=z\right]  \leq 0.
    \end{aligned}
\end{equation} 
The representation of the testable implication in Equation \eqref{eq:implication3} is now easier to test thanks to the recent advances on testing conditional moment inequalities. 
%\subsubsection{Testing using Intersection Bounds Framework}
More precisely, \cite{chernozhukov2013intersection} propose a method that allows us to test a set of conditional moment inequalities in the form of Equation \eqref{eq:implication3}, which are often called intersection bounds. %We use the testable implication described in Eq.(\ref{eq:test_implication}) as an example to illustrate our testing procedure.

To test inequalities in condition \eqref{eq:implication3}, we will first plug in the $\sqrt{n}$-consistent estimators $(\hat{\beta}_0,\hat{\beta}_1)$ for $(\beta_0,\beta_1)$. To simplify notation, we define $W_1 = Y - \hat{\beta}_0-\hat{\beta}_1 X$ and $W_2 = - Y + \hat{\beta}_0+\hat{\beta}_1 X$, and a set $\mathcal{V}=\{(z, j): z \in \mathcal{Z}, j \in\{1,2\}\}$. Then, we denote  $\theta(v) \equiv \mathbb{E}\left[W_j \mid Z=z\right]$ for $v=(z, j)$. 
We focus on testing the hypothesis
\begin{equation}
    H_0: \theta_0 \equiv \sup _{v \in \mathcal{V}} \theta(v) \leq 0 \quad \text { v.s. } \quad H_1: \theta_0>0.
    \label{eq:hypo:inequalities}
\end{equation}
Second, to test the null hypothesis $H_0$ against the alternative $H_1$ in Equation \eqref{eq:hypo:inequalities}, we need to estimate the supremum statistic $\theta_0$, which requires estimating the conditional moment $\theta(v)$ for each $v$. To eliminate the first step plug-in estimation bias asymptotically, we recommend estimating the conditional moment $\theta(v)$ nonparametrically. Letting $\hat{\theta}(v)$ denote an estimator for $\theta(v)$, we show in Appendix \ref{app:validity_plug_in} that all asymptotic properties are preserved if we then plug in a $\sqrt{n}$-consistent estimator $(\hat{\beta}_0, \hat{\beta}_1)$ for $(\beta_0, \beta_1)$ in the first step. If we simply take supremum over $\hat{\theta}(v)$ and use it as an estimator for $\theta(v)$, we would have a finite-sample bias because of estimation errors. To correct this bias, \cite{chernozhukov2013intersection} propose a precision-corrected estimator\footnote{See Appendix \ref{app:precision-correct} for how to choose a precision-corrected estimator.} for $\theta_0$,
\begin{equation} \label{sup_estimator}
    \hat{\theta}_{1-\alpha} \equiv \sup _{v \in \mathcal{V}}\left\{\hat{\theta}(v)-k_{1-\alpha} \hat{s}(v)\right\}.
    %\label{clrtest}
\end{equation}
The precision correction term is $k_{1-\alpha} \hat{s}(v)$. Here, $\hat{s}(v)$ is the standard error of $\hat{\theta}(v)$, $k_{1-\alpha}$ is $(1-\alpha)$-quantile of an approximated distribution of $\sup _{v \in \mathcal{V}} \{(\hat{\theta}(v)-\theta(v))/ \sigma(v)\}$, $\sigma(v)$ is the standard deviation of $\hat{\theta}(v)$. Heuristically speaking, the precision correction term adjusts the estimator with estimation error uniformly on $\mathcal{V}$.

The estimator can be used as a test statistic for testing the hypothesis (\ref{eq:hypo:inequalities}). We reject the null hypothesis at the significance level $\alpha$ if $\hat{\theta}_{1-\alpha} > 0$. We use the \texttt{clrtest/clrbound} Stata commands from \cite{chernozhukov2015implementing} to test the hypothesis in (\ref{eq:hypo:inequalities}).

\subsection{Extension to General Parametric Separable Models}
\subsubsection{Linear models with multiple regressors}\label{sec:multilinear}
%\subsubsection{Testing Identifying Assumptions in Linear Models}
Consider a linear IV model with multiple (potentially endogenous) regressors where $X$ is a vector instead of a scalar:
\begin{equation} \label{model:ols}
    Y=X'\beta+U.
\end{equation}
Let $Z=(Z_1, Z_2, \ldots, Z_k)'$ be a candidate IV for $X$. To identify and estimate the parameters $\beta$ in Equation (\ref{model:ols}), researchers usually impose the following assumptions.

\begin{assumption}(IV conditions) \label{ass:ivols}
    \begin{enumerate}
        \item (Exogeneity assumption) $\mathbb{E}[U \mid Z]=0$ almost surely.
        \item (Rank condition) The matrix $\mathbb E[ZX']$ is invertible.
        %\item (Exclusion restriction) 
    \end{enumerate}
\end{assumption}

% For clarity, let $k=2$ and rewrite Eq.(\ref{model:ols}) as $Y=\beta_0+\beta_1X+U$. Although the vector $X$ in Assumption \ref{ass:ols} is now $(1,X)'$ with the scalar random variable $X$, the idea can be easily generated to the case with a vector random variable $X$. Under Assumption \ref{ass:ols}, we can first identify  
Under the IV conditions in Assumption \eqref{ass:ivols},
\begin{equation*}
    \beta= \mathbb E[ZX']^{-1}\mathbb E[ZY].
\end{equation*}
As a result, $U=Y-X'\beta=Y-X'\mathbb E[ZX']^{-1}\mathbb E[ZY]$ where each term in $U$ is observable or identified from the data. Under the exogeneity condition in Assumption \ref{ass:ivols},
\begin{equation}
    \mathbb{E}[U \mid Z] = \mathbb{E}\left[Y-X'\mathbb E[ZX']^{-1}\mathbb E[ZY] \mid Z=z\right] =0\ \text{ for all } z,
    \label{test:ols1}
\end{equation}
which is equivalent to a set of joint conditional moment inequalities in Equation (\ref{test:ols3}) and consistent with the framework of intersection bounds.

\begin{proposition}[Testable implications of linear IV models] \label{thm:ols}
    Consider the linear model in Eq. \eqref{model:ols}. Then, testable implications of Assumption \ref{ass:ivols} are given by
    \begin{equation}
        \begin{aligned}
            & \sup_{z\in\mathcal{Z}} \mathbb E \left[Y- X' \beta^* \mid Z=z\right] \leq 0,\\
            & \sup_{z\in\mathcal{Z}} \mathbb E \left[-Y+X' \beta^* \mid Z=z\right] \leq 0,
        \end{aligned}
        \label{test:ols3}
    \end{equation}
    where $\beta^*=\mathbb E [ZX']^{-1} \mathbb E [ZY]$ is identified under Assumption \ref{ass:ivols}.
\end{proposition}
The testable implications in Equation \eqref{test:ols3} are sharp, and the proof is similar to that of Proposition \ref{prop1} and is therefore omitted.

We can apply the same methods to test other identifying assumptions that are commonly used in the literature on linear IV models, for example, the homoskedasticity assumption $\mathbb E[U^2 \mid Z]=\sigma^2$. %Corollary \ref{corr:ols1} states the testable implications if we add the homoskedasticity assumption.
% \begin{corollary}(Testable implications of OLS assumptions with homoskedasticity)\label{corr:ols1}
The testable implications for the linear IV model in Equation \eqref{model:ols} and Assumption \ref{ass:ivols} coupled with the homoskedasticity assumption are
    \begin{equation}
    \begin{aligned}
        & \sup _{Z \in \mathcal{Z}} \mathbb{E}\left[Y-X' \beta^* \mid Z=z\right] \leq 0, \\
        & \sup _{z \in \mathcal{Z}} \mathbb{E}\left[-Y+X' \beta^* \mid Z=z\right] \leq 0, \\
        & \sup _{z \in \mathcal{Z}} \mathbb{E}\left[\left(Y-X' \beta^*\right)^2-{\sigma^*}^2 \mid Z=z\right] \leq 0, \\
        & \sup _{z \in \mathcal{Z}} \mathbb{E}\left[-\left(Y-X' \beta^*\right)^2+{\sigma^*}^2 \mid Z=z\right] \leq 0.
    \end{aligned}
    \label{test:olshomo}
    \end{equation}
    %where $(\beta^*, \sigma^*)$ is identified under Assumption \ref{ass:ols} and Assumption \ref{ass:olshomo}.

%     Similarly, the testable implications for the linear model in Eq.(\ref{model:ols}) with Assumption \ref{ass:iv} and Assumption \ref{ass:olshomo} are 
%     \begin{equation}
%         \begin{aligned}
%             & \sup _{z \in \mathcal{Z}} \mathbb{E}\left[Y-X' \beta^* \mid Z = z\right] \leq 0, \\
%             & \sup _{z \in \mathcal{Z}} \mathbb{E}\left[-Y+X' \beta^* \mid Z = z\right] \leq 0, \\
%             & \sup _{z \in \mathcal{Z}} \mathbb{E}\left[\left(Y-X' \beta^*\right)^2-{\sigma^*}^2 \mid Z = z\right] \leq 0, \\
%             & \sup _{z \in \mathcal{Z}} \mathbb{E}\left[-\left(Y-X' \beta^*\right)^2+{\sigma^*}^2 \mid Z = zx\right] \leq 0.
%         \end{aligned}
%         \label{test:ivhomo}
%         \end{equation}
%         where $(\beta^*, \sigma^*)$ is identified under Assumption \ref{ass:iv} and Assumption \ref{ass:olshomo}.
% \end{corollary}

\subsubsection{Nonlinear parametric separable models}
Consider a more general model where $Y = m(X, \theta_0) + U$, where $\mathbb{E}\left(U \mid Z\right) = 0$, $m(\cdot, \theta)$ defined on $\mathbb{R}^{k} \times \Theta$ is a known function, and $\Theta$ is a parameter space. This model can generally be tested without assuming that there is a unique parameter $\theta_0$ such that $\mathbb E[Y-m(X,\theta_0) \mid Z]=0$. Instead, we can characterize the identified set for $\theta_0$ as:  
\begin{eqnarray*}
    \Theta_{I}=\bigg\{ \theta \in \Theta: \sup_{z \in \mathcal Z} \mathbb E[Y-m(X,\theta) \mid Z=z] \leq 0\ \text{ and }\ \sup_{z\in \mathcal Z} \mathbb E[-Y+m(X,\theta) \mid Z=z] \leq 0 \bigg\}.
\end{eqnarray*}
The above identified set can be computed through a grid search using the intersection bounds framework by \cite{chernozhukov2013intersection}. As a byproduct, if the identified set $\Theta_I$ is empty, then the model is rejected. 

One can add other commonly used assumptions such as monotonicity of $m(\cdot,\theta)$ in $\theta$ for all values of $X$ (to ensure uniqueness of $\theta_0$), homoskedasticity, etc. to tighten the identified set.

\begin{example}[Box-Cox regression model] For any $x > 0$, define
\begin{equation*}
    x^{(\lambda)} \equiv \begin{cases}\frac{x^\lambda-1}{\lambda}, & \text { if } \lambda \neq 0 \\ \log (x), & \text { if } \lambda=0\end{cases}
\end{equation*}
and the Box-Cox model is specified as
\begin{equation}
    Y=\beta_0+\beta_1 X^{(\lambda)}+U.
    \label{model:boxcox}
\end{equation}

In the Box-Cox model, the nonlinear function $m(\cdot)$ takes the form $m(X, \theta)=\beta_0+\beta_1 X^{(\lambda)}$, and the parameter $\mathbb{\theta} = \left(\beta_0, \beta_1, \lambda\right)$.
\end{example}

\begin{example}[CES Function]
Another typical example discussed in \cite{hansen2022econometrics} is the constant elasticity of substitution (CES) function, which is a generalization of the Cobb-Douglas production function. Researchers often assume the CES structure to estimate production functions. The CES function with two inputs, $X_1, X_2$, is
\begin{equation*}
    Y=\left\{\begin{array}{cc}
        A\left(\alpha X_1^\lambda+(1-\alpha) X_2^\lambda\right)^{\nu / \lambda}, & \text { if } \lambda \neq 0 \\
        A\left(X_1^\alpha X_2^{(1-\alpha)}\right)^\nu, & \text { if } \lambda=0
    \end{array}\right.
\end{equation*}
We can take the logarithm of Y and set $\log A = \beta_0 + U$, where $U$ represents the unobserved productivity. Then the CES function implies the following nonlinear model,
\begin{equation}
    \log Y=\left\{\begin{array}{cl}
        \beta_0+\frac{\nu}{\lambda} \log \left(\alpha X_1^\lambda+(1-\alpha) X_2^\lambda\right)+U, & \text { if } \lambda \neq 0 \\
        \beta_0 + \nu \left(\alpha \log X_1 + (1-\alpha) \log X_2\right) + U, & \text { if } \lambda=0,
    \end{array}\right.  
    \label{model:ces}
\end{equation}
where the parameter $\mathbb{\theta} = (\lambda, \nu, \alpha, \beta_0)$, and the function $m(X, \mathbb{\theta})$ is nonlinear in $\mathbb{\theta}$.
\end{example}

\subsubsection{Nonlinear parametric nonseparable models}
Although we focus on parametric separable models in this paper, we explain in Example~\ref{ex:probit} below how our proposed approach could be extended to some nonseparable models such as probit models.

\begin{example}[Probit with endogeneity]\label{ex:probit}
Consider the following parametric model
\begin{eqnarray*}
\left\{\begin{array}{clc}
        Y &=& \mathbbm{1}\{X\beta -U \geq 0\}, \\
        X &=& Z \delta + V,
    \end{array}\right.
\end{eqnarray*}
    where $(U,V)'$ is jointly normally distributed with coefficient of correlation $\rho$, and $Z \indep (U,V)$. We can write $U=\rho V + \varepsilon$, where $\varepsilon \indep (V,Z),$ and $\varepsilon \sim N(0,1-\rho^2)$. The coefficients $(\beta, \delta, \rho)$ are identified and can be estimated through the maximum likelihood method. We can write
    \begin{eqnarray*}
        \mathbb E[Y|X=x, Z=z] &=& \mathbb P(U \leq X \beta | X=x, Z=z),\\
        &=& \mathbb P(\rho V +\varepsilon \leq X \beta | X=x, Z=z),\\
        &=& \mathbb P(\varepsilon \leq x \beta - \rho(x-z \delta) | X=x, Z=z),\\
        &=& \Phi\left(x\frac{\beta-\rho}{\sqrt{1-\rho^2}}+z \frac{\rho \delta}{\sqrt{1-\rho^2}}\right),
    \end{eqnarray*}
    which implies $ \mathbb E\left[Y-\Phi\left(X(\beta-\rho) /\sqrt{1-\rho^2 }+Z \rho \delta / \sqrt{1-\rho^2}\right) | X=x, Z=z\right]=0$ for all $(x,z)$. Since the parameters $(\beta, \delta, \rho)$ are identified, this latter implication is testable.     
\end{example}

This example above can be generalized in the following way. Suppose the observed vector $(Y,X,Z)$ satisfies
\begin{eqnarray*}
\left\{\begin{array}{clc}
        Y &=& g(X,\beta,U), \\
        X &=& h(Z,\delta) + V,
    \end{array}\right.
\end{eqnarray*}
where $g$ and $h$ are known functions, $U=\rho V+\varepsilon,$ $\varepsilon \indep V$, and $\varepsilon | X,Z \sim F_{\varepsilon}$ (known). This model is parametric and is separable in the ``first-stage'' equation for $X$, even though it is nonseparable in the outcome equation.

Suppose also that $(\beta, \delta, \rho)$ are identified. Then, we have a generalized version of the testable implication 
$$\mathbb E\left[Y-\int g(X,\beta, e+\rho(X-h(Z,\delta))d F_{\varepsilon}(e) | X=x,Z=z\right]=0,$$ which can be converted into conditional moment inequalities
\begin{eqnarray*}
\sup_{x,z}\mathbb E\left[Y-\int g(X,\beta, e+\rho(X-h(Z,\delta))d F_{\varepsilon}(e) | X=x,Z=z\right] &\leq& 0,\\
\sup_{x,z}\mathbb E\left[-Y+\int g(X,\beta, e+\rho(X-h(Z,\delta))d F_{\varepsilon}(e) | X=x,Z=z\right] &\leq& 0.
\end{eqnarray*}

\subsubsection{Semiparametric linear single-index models}
Our proposed approach can also be extended to binary response models that make no distributional assumptions on the error term, but restrict the dependence on covariates to take a parametric form such as a linear single-index model:

\begin{example}[Binary response with linear index]\label{ex:linearindex}
Consider the model
$$Y = \mathbbm{1}\{X'\beta -U > 0\},$$
\end{example}
\noindent where $X \indep U$ but the distribution $G(u)$ of $U$ is left unrestricted. As pointed out by \citet{ichimura1993}, identification of $\beta$ in this model is possible up to a normalization provided some regularity conditions hold, for example, that $G$ is differentiable and at least one regressor is continuously distributed.

Note that this model implies that
$$\mathbb{E}[Y|X=x] = \mathbb{P}(U < x'\beta) = G(x'\beta),$$
which carries the testable restriction that for any $x_1$ and $x_2$ such that $x_1'\beta=x_2'\beta$, $\mathbb{E}[Y|X=x_1]=\mathbb{E}[Y|X=x_2]$. If the support of $X$ is a known subset $\mathcal{X} \in \mathbb{R}^k$, we can therefore test this model through the equality restriction that
\begin{equation} \label{eq:linindexrestriction}
   % \inf_{b \in \mathbbm{R}^k} \left\{ 
    \sup_{x_1,x_2 \in \mathcal{X}} \left|\mathbb{E}[Y|X=x_1]-\mathbb{E}[Y|X=x_2]\right|\cdot \mathbbm{1}(x_1'b = x_2'b) 
    %\right\} 
    = 0.
\end{equation}
%{\color{blue} Should we use the absolute value sign outside the difference of conditional means?
%\begin{equation*}
%    \inf_{b \in \mathbbm{R}^k} \left\{ \sup_{x_1,x_2 \in \mathcal{X}} \left|\mathbbm{E}[Y|X=x_1]-%\mathbbm{E}[Y|X=x_2]\right| \cdot \mathbbm{1}(x_1'b = x_2'b) \right\} = 0.
%\end{equation*}
%}
We know that when $b=\beta^*$ the inner supremum is equal to zero, where $\beta^*$ is the true value of $\beta$. To bring the above into the intersection bounds framework, make use of the fact that $\beta^*$ is identified, and rewrite \eqref{eq:linindexrestriction} as
$$\sup_{(x_1,x_2) \in \mathcal{V}(\beta^*)} \left\{\mathbb{E}[Y|X=x_1]-\mathbb{E}[Y|X=x_2]\right\} \le 0,$$
$$\sup_{(x_1,x_2) \in \mathcal{V}(\beta^*)} \left\{\mathbb{E}[Y|X=x_2]-\mathbb{E}[Y|X=x_1]\right\} \le 0,$$
% {\color{blue} Is it better to write the inequalities as follows to keep them consistent with forms in \cite{chernozhukov2013intersection}?
% \begin{equation*}
%     \begin{aligned}
%         &\sup_{x_2 \in \mathcal{V}(x_1;\beta^*)} \left\{\mathbbm{E}\left[\mathbbm{E}[Y|X=x_1]-Y|X=x_2\right]\right\} \le 0 \\
%         &\sup_{x_2 \in \mathcal{V}(x_1;\beta^*)} \left\{\mathbbm{E}\left[Y - \mathbbm{E}[Y|X=x_1]|X=x_2\right]\right\} \le 0
%     \end{aligned}
% \end{equation*}
% for all $x_1 \in \mathcal{X}$, where $\mathcal{V}(x_1;\beta^*)$ is  defined as $\mathcal{V}(x_1;\beta^*) := \{x_2 \in \mathcal{X}: x_1'b = x_2'b\}$ given $x_1$ and identified $\beta^*$.
% }
i.e., that $\mathbb{E}[Y|X=x_1]-\mathbb{E}[Y|X=x_2]=0$ for all $(x_1,x_2) \in \mathcal{V}(b)$, where
we define $\mathcal{V}(b) := \{(x_1,x_2) \in \mathcal{X} \times \mathcal{X}: x_1'b = x_2'b\}$.
The results of \citet{chernozhukov2013intersection} require the set over which the supremum or infimum is taken across be compact, so we must maintain this assumption here. A sufficient condition is that $\mathcal{X}$ be compact in $\mathbb{R}^k$. 

%$$\left\{ \sup_{x_1,x_2 \in \mathcal{X}} \left\{\mathbbm{E}[Y|X=x_1]-\mathbbm{E}[Y|X=x_2]\right\}\cdot \mathbbm{1}(x_1'\beta^* = x_2'\beta^*) \right\} \le 0$$
We illustrate how sample splitting can help transform the inequalities above into conditional moment inequalities. Suppose we have an i.i.d. sample $\{Y_i,X_i\}_{i=1}^{n}$, and we randomly split the sample into two sub-samples $\{Y_i^{(1)},X_i^{(1)}\}_{i=1}^{n_1}$ and $\{Y_i^{(2)},X_i^{(2)}\}_{i=n_1+1}^{n}$. Then the inequalities above are equivalent to the following:
$$\sup_{(x_1,x_2) \in \mathcal{V}(\beta^*)} \left\{\mathbb{E}[Y^{(1)}-Y^{(2)}|X^{(1)}=x_1, X^{(2)}=x_2]\right\} \le 0,$$
$$\sup_{(x_1,x_2) \in \mathcal{V}(\beta^*)} \left\{\mathbb{E}[Y^{(2)}-Y^{(1)}|X^{(1)}=x_1, X^{(2)}=x_2]\right\} \le 0.$$

\section{Monte Carlo Simulations}\label{sec: MonteCarlo}

\subsection{Size of the test}

First, we generate a linear model with an instrument variable that satisfies Assumption \ref{ass:mean_indep}. The model takes the form
\begin{equation}
    \left\{ \begin{array}{lcl}
     Y_i &=&  \beta_0 + \beta_1 X_i +U_i \\
     X_i &=& \gamma_0 + \gamma_1 Z_i + V_i
     \end{array} \right.
     \label{simu:sizeolsiv}
\end{equation}
where $\beta_0 = 0$, $\beta_1 = 2$, $\gamma_0 = 0$ and $\gamma_1 = 3$. We have done $500$ Monte Carlo replications. For each replication, we set the sample size to be $n$. 
We randomly draw $X_i \stackrel{\text { i.i.d. }}{\sim} \mathcal{U}[-3, 3]$, $(U_i, V_i)$ i.i.d. from $\mathcal{N}(0, \Sigma)$ with $\Sigma=(1, 0.5; 0.5, 2)$, and draw $Z_i \stackrel{\text { i.i.d. }}{\sim} \mathcal{U}[-3, 3]$ independent with $(U_i, V_i)$, $i = 1, \dots, n$. Table \ref{tab:sizeolsiv} presents the simulation results of testing the above DGP under Assumption \ref{ass:mean_indep}.
\begin{table}[!ht] 
    \centering
    \caption{Rejection rates when instrument identifying assumption holds}
    \begin{tabular}{l|ccc}
    & \multicolumn{3}{c}{Series estimation}   \\ 
    Significance level &$10\%$&$5\%$  &$1\%$\\  
    \midrule
    \midrule
    $n=200$ &$45.4\%$&$36.4\%$&$24.8\%$  \\ 
    $n=500$ &$25.8\%$&$19.4\%$&$9.0\%$   \\
    $n=1000$ &$17.2\%$&$13.2\%$&$7.4\%$    \\   
    $n=2000$ &$15.0\%$&$10.4\%$&$3.4\%$    \\
    $n=3000$ &$8.0\%$&$4.6\%$&$1.4\%$  \\ \hline
    \end{tabular}%
    \begin{center}
    \footnotesize{Based on 500 replications. Use series regression to estimate the conditional expectations.}
    \end{center}
    \label{tab:sizeolsiv}
\end{table}

To implement the test, we first estimate $\hat{\beta}_0$ and $\hat{\beta}_1$ under Assumptions \ref{ass:mean_indep} and \ref{ass:relevance} and obtain $\hat{U}_i = Y_i - \hat{\beta}_0 - {\beta}_1 X_i$. Then we use the Stata package of \cite{chernozhukov2015implementing} to compute the rejection rate of the testable implications in Equation \eqref{eq:implication3}. To estimate the supremum test statistic, we need to select grid points over the support of $Z_i$. We take $100$ grid points uniformly over the 1 centile to 99 centiles of the support of $Z_i$. Also, we select the series estimation to estimate conditional means $\mathbb{E}[\hat{U}_i \mid Z_i]$. 

From the results in Table \ref{tab:sizeolsiv}, we can observe that the rejection rates are higher than nominal sizes when the sample sizes are small. The reason is that we use a nonparametric method to estimate the conditional expectations, which have large estimation errors when the sample sizes are small. As the sample size increases, the rejection rates drop at each significance level and finally are close to nominal sizes, which is consistent with our expectations. When the sample size increases to 3000, rejection rates fall below the nominal sizes at significance levels of $10\%$ and $5\%$. The low rejection rates suggest that our tests might be conservative, which is caused by the precision correction term when constructing the test statistics.

We also come up with the null hypothesis specifications under the exogenous assumption of the regressor $X$ and the homoskedasticity assumption. The results are presented in Appendix \ref{app:simulation}, and we can observe a similar pattern for rejection rates as in Table \ref{tab:sizeolsiv}.

Next, we consider a nonlinear model that satisfies the IV assumption. We generate a DGP in the form of the Box-Cox model described in Equation \eqref{model:boxcox}.

\begin{equation}
    \left\{\begin{array}{l}
          Y_i = \beta_0 + \beta_1 X_i^{(\lambda)}+U_i, \\
          X_i = \gamma_0 + \gamma_1 Z_i+V_{+i},
          \end{array}\right.
          \label{simu:sizenlsiv}
\end{equation}
where $\beta_0 = 0$, $\beta_1 = 2$, $\gamma_0 = 0$, $\gamma_1 = 2$, and
\begin{equation}
        X_i^{(\lambda)} = \begin{cases}\frac{X_i^\lambda-1}{\lambda}, & \text { if } \lambda \neq 0 \\ \log (X_i), & \text { if } \lambda=0\end{cases}
\end{equation}
for $\lambda = 0, -1, 1$. We randomly draw $(U_i, V_i)$ i.i.d. from $\mathcal{N}(0, \Sigma)$ with $\Sigma=(1, 0.5; 0.5, 2)$, and $V_{+i}$ denotes the positive part of $V_i$. We also draw $Z_i \stackrel{\text { i.i.d. }}{\sim} \mathcal{U}(0, 10]$ independent with $(U_i, V_i)$. 
% We estimate $\hat{\beta}_0$ and $\hat{\beta}_1$ under the nonlinear IV condition and obtain $\hat{U}_i = Y_i - \hat{\beta}_0 - \hat{\beta}_1 X_i^{(\lambda)}$. We choose 100 grid points uniformly over the support of $Z_i$. We also choose the series estimation to estimate $\mathbb{E}[\hat{U}_i \mid Z_i]$. Table \ref{tab:sizenls} present the simulation results. 
The simulation results are presented in Table \ref{tab:sizenlsiv}.
\begin{table}[!htbp]
  \centering
  \caption{Rejection rates when IV assumptions hold in a nonlinear model}
    \begin{tabular}{r|l|rrr}
    \multicolumn{1}{c|}{\multirow{2}[1]{*}{$\lambda$}} & \multicolumn{1}{c|}{\multirow{2}[1]{*}{Sample Size}} & \multicolumn{3}{c}{Significance level} \\
          &       & \multicolumn{1}{c}{10\%} & \multicolumn{1}{c}{5\%} & \multicolumn{1}{c}{1\%} \\
    \midrule
    \midrule
    0     & n = 200 & 16.00\% & 11.80\% & 7.00\% \\
          %& n = 500 & 8.00\% & 4.00\% & 1.20\% \\
          & n = 1000 & 8.80\% & 5.00\% & 1.40\% \\
          & n = 2000 & 5.00\% & 2.20\% & 0.80\% \\
          & n = 3000 & 5.80\% & 2.80\% & 0.40\% \\
    \midrule
    -1    & n = 200 & 23.80\% & 19.80\% & 12.40\% \\
          %& n = 500 & 15.80\% & 10.40\% & 3.80\% \\
          & n = 1000 & 15.40\% & 10.80\% & 5.80\% \\
          & n = 2000 & 15.40\% & 10.60\% & 4.80\% \\
          & n = 3000 & 15.40\% & 9.60\% & 2.60\% \\
    \midrule
    1     & n = 200 & 15.00\% & 11.20\% & 6.40\% \\
          %& n = 500 & 8.00\% & 3.20\% & 1.60\% \\
          & n = 1000 & 6.20\% & 2.40\% & 0.20\% \\
          & n = 2000 & 4.40\% & 1.60\% & 0.00\% \\
          & n = 3000 & 4.00\% & 1.80\% & 0.20\% \\
    \bottomrule
    \end{tabular}%
    \begin{center}
        \footnotesize{Based on 500 replications. Use series regression to estimate the conditional expectations.}
    \end{center}
  \label{tab:sizenlsiv}%
\end{table}%

When the sample size is $200$, the rejection rates at each significance level are larger than the nominal size for each value of $\lambda$. However, when the sample size increases, the rejection rates at each significance level decrease and are finally below the nominal size. %Also, rejection rates fall below the nominal size when the sample size increases. 
Again, this shows that our tests are conservative, and researchers need to pay attention to this in practice. Note that the rejection rates are higher for $\lambda=-1$ compared to $\lambda=0,1$, and they do not fall below the nominal size for $\lambda=-1$ when $n=3000$ as they do for $\lambda=0,1$.

% We also propose a specification that satisfies Assumption \ref{ass:nliv} and presents the results in Appendix \ref{app:size}. The pattern of rejection rates in Table \ref{tab:sizenls} preserves in the test of the specification under Assumption \ref{ass:nliv}. 

\subsection{Power of the test} \label{sec:power}

To investigate the power properties of our testing procedure, we consider the similar alternatives proposed in \cite{horowitz2001adaptive} and \cite{tripathi2003testing}. We first consider a specification of a linear model that violates the IV exogeneity assumption (Assumption \ref{ass:mean_indep}). In the model specification \eqref{simu:sizeolsiv}, we generate $Z_i$, $i = 1, \ldots, 1000$, from $Z_i \stackrel{\text { i.i.d. }}{\sim} \mathcal U{[-3,3]}$. Then, we generate $U_i = L / \sigma \cdot \phi(Z_i / \sigma) + \tilde{U}_i$, where $\tilde{U}_i = \min\{\max\{-3, \tilde{V}_i\}, 3\}$, $(\tilde{V}_i, V_i)$ is drawn independently of $Z_i$ and i.i.d. from $\mathcal{N}(0, \Sigma)$ with $\Sigma=(1, 0.5; 0.5, 1)$.\footnote{Note that $\tilde{U}_i$ has mean zero, as $\tilde{U}_i$ is a mean zero normal distribution truncated on $[-3,3]$ so that its distribution is symmetric around 0.} Under the constructed DGP, Assumption \ref{ass:mean_indep} fails. The conditional mean function $\mathbb{E}[U_i \mid Z_i = z] =  L / \sigma \cdot \phi(z / \sigma)$ is smooth in $z$ and has a unique maximizer. Under this constructed DGP, $L$ and $\sigma$ are two constants that determine the shape of $\mathbb{E}[U_i \mid Z_i = z]$. $L$ measures the average deviation of the conditional mean from zero, where a larger $L$ indicates a greater average deviation over the support. Also, $\sigma$ determines the shape of this function: a smaller $\sigma$ implies that the function $\mathbb{E}[U_i \mid Z_i =z]$ is more peaked around the maximizer. We conduct $500$ replications for each value of $L$ and $\sigma$. Table \ref{tab:power_iv} presents the rejection rates.

\begin{table}[htbp]
    \centering
    \caption{Rejection rates when instrument assumptions fail in linear model}
      \begin{tabular}{c|c|rrr}
            &       & \multicolumn{3}{c}{Significance Level} \\
      $L$     & $\sigma$ & \multicolumn{1}{c}{10\%} & \multicolumn{1}{c}{5\%} & \multicolumn{1}{c}{1\%} \\
      \midrule
      \midrule
      \multirow{4}[2]{*}{0.1} & 1     & 17.80\% & 11.00\% & 4.60\% \\
            & 0.5   & 19.60\% & 11.40\% & 4.40\% \\
            & 0.25  & 23.60\% & 14.40\% & 5.20\% \\
            & 0.1   & 25.80\% & 17.60\% & 6.60\% \\
      \midrule
      \multirow{4}[2]{*}{0.5} & 1     & 40.00\% & 30.00\% & 12.00\% \\
            & 0.5   & 86.20\% & 79.20\% & 58.20\% \\
            & 0.25  & 100.00\% & 99.60\% & 99.00\% \\
            & 0.1   & 100.00\% & 99.80\% & 99.60\% \\
      \midrule
      \multirow{4}[2]{*}{1} & 1     & 87.40\% & 78.80\% & 51.60\% \\
            & 0.5   & 100.00\% & 100.00\% & 100.00\% \\
            & 0.25  & 100.00\% & 100.00\% & 100.00\% \\
            & 0.1   & 100.00\% & 100.00\% & 100.00\% \\
      \bottomrule
      \end{tabular}%
    \label{tab:power_iv}%
    \begin{center}
        \footnotesize{Based on 500 replications. Use series regression to estimate conditional expectations.}
    \end{center}
\end{table}%
  
In Table \ref{tab:power_iv}, we can observe that when $L$ increases, which suggests that the deviation of the conditional mean $\mathbb E[U_i \mid Z_i=z]$ from zero is more significant, the rejection rates of our test increase at each significance level. This is consistent with our conjectures that our test is more powerful when deviations of the null are more significant. In addition, we can observe that given the value of $L$, the rejection rates increase when $\sigma$ becomes smaller. %The reason is that we are using a sup norm to measure the deviation from the null instead of an $L^2$ metric. Even if the average deviation stays the same, when the function $\mathbb{E}[U_i \mid Z_i = z]$ becomes more peak, our test is more powerful in detecting the violations. 
This relates to the precision-correction term when constructing the supremum test statistics. As we mentioned in Section \ref{sec: analytical_framework}, if the conditional moment function is more peaked, the method of \cite{chernozhukov2013intersection} can correct errors of estimating supremum test statistic more accurately. Therefore, when $\sigma$ becomes smaller, the conditional mean function is more peaked, and our test with the supremum test statistic can achieve greater power. %Even when $L = 0.5$, where the deviation from the null is insignificant, the rejection rates are still high at each significance level if the conditional mean function is nonflat (when $\sigma \geq 0.5$). However, if using tests that measure deviations from the null with an $L^2$ norm, they may not have the power against such alternatives since they measure the difference with average squared deviations over the full support of $Z$. 
We also get the rejection rates for specifications that violate the exogenous assumption on $X$. The simulation results are in Appendix \ref{app:simulation} and display a similar pattern as in Table~\ref{tab:power_iv}.

To compare our testing method with the common overidentification test, we consider the same DGP with $L = 0.5, \sigma = 0.25$, and we convert the conditional moment restriction in Assumption \ref{ass:mean_indep} to an unconditional moment condition $\mathbb{E}[h(Z_i) U_i] = 0$. We choose $h(Z) = (Z, Z^2, Z^3)'$ and estimate the parameters $(\beta_0, \beta_1)$ with two-stage least square method. Then, we apply Sargan's test to check the instrument mean independence validity of $Z$. We conduct $500$ simulations with different sample sizes and plot the rejection rates of those two tests in Figure \ref{fig:power_curve_linear}.
\begin{figure}[!htt]
        \centering
        \caption{Power curves for testing IV assumptions in linear model}
        \includegraphics[width = 0.8\textwidth]{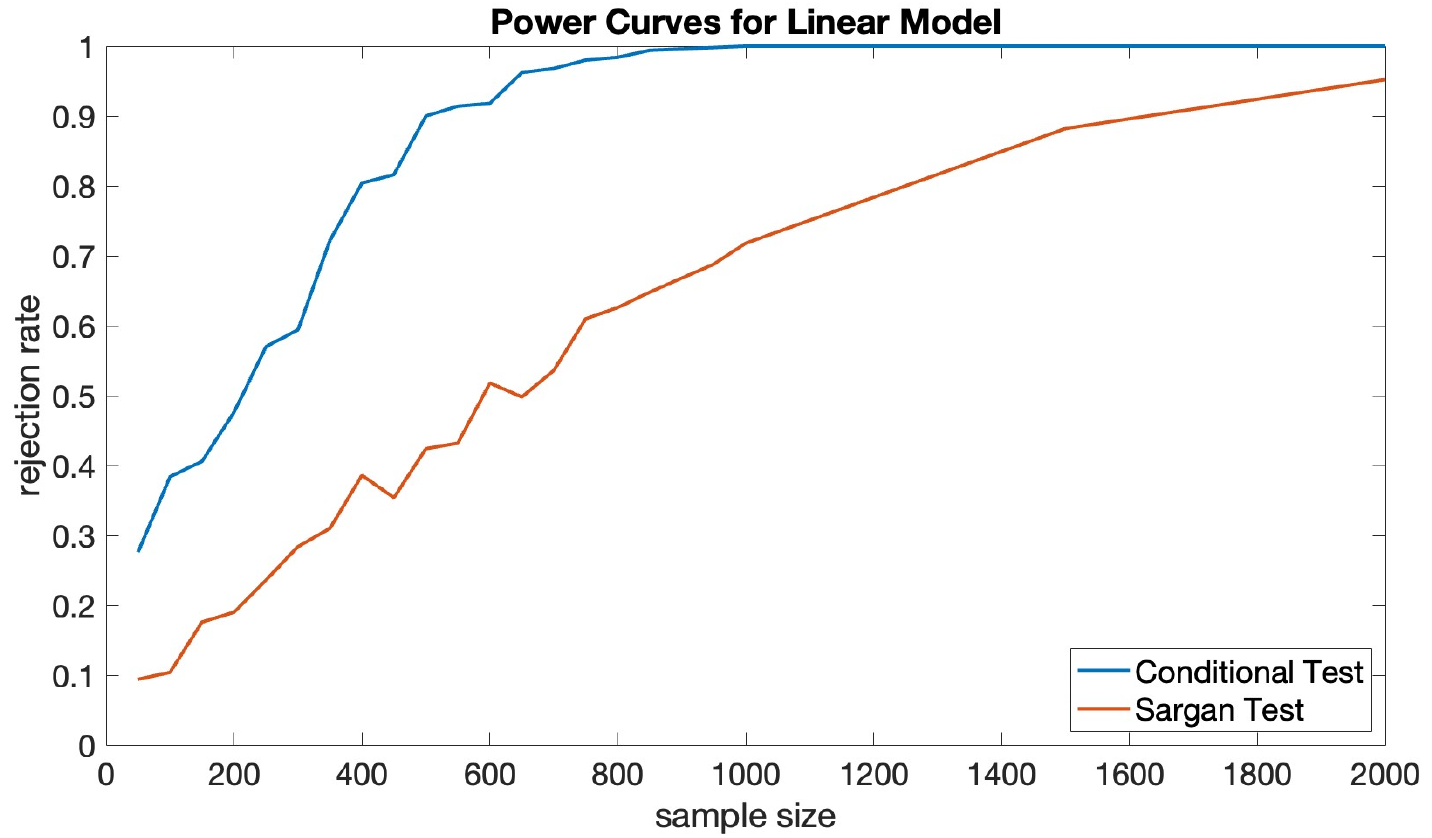}
        \label{fig:power_curve_linear}
\end{figure}
This Figure shows that if the instrument mean independence assumption is violated in the linear model, the rejection rates of both tests increase with the sample size and approach 1 as the sample size grows, which suggests that both tests are consistent. Also, our proposed conditional moment inequality test gives a higher rejection rate than Sargan's test for any sample size. Therefore, our test achieves better performance in power, and we suggest that researchers try our testing procedure to test the instrument validity even if they use the overidentification test and cannot reject the IV assumption. We also compare the power of our test with Sargan's test for other values of $L$ and $\sigma$, and the results are presented in the Appendix \ref{app:comparision} and show that our test can achieve higher powers than Sargan's test in most cases.

\section{Solutions when the parametric IV model is rejected} \label{sec:id_relax_assumption}
Several reasons may lead to the rejection of the identifying assumptions in the IV model. The rejection could be due to the violation of the functional form assumptions (misspecification, heterogeneous effects, etc.), the invalidity of the instrumental variable (exogeneity and/or exclusion restriction), etc. The literature has offered many solutions. \cite{nevo2012identification} develop an identification result in the parametric model when the correlation between the instrument and the error term has the same sign as the correlation between the regressor and the error term. They relax the IV exogeneity assumption while maintaining the exclusion restriction. \cite{Conley2012} present methods for performing inference while allowing for violations of the exclusion restriction but maintain exogeneity of the IV. In this framework, \cite{Masten2021} recommend reporting the set of parameters that are consistent with minimally nonfalsified models, which they refer to as the falsification adaptive set. So, when the IV model is rejected, we recommend that the researcher explores the possible relaxations discussed in the above mentioned papers if s/he believes in the functional form assumption. Below, we propose a way to relax the functional form while maintaining the exogeneity and exclusion restriction assumptions. 

%One reason is that the imposed functional assumptions could be wrong. Specifically, it may be incorrect to assume $\mathbb{E}[Y \mid X]$ to be a known parametric function $m(X, \beta)$ of $X$. In this section, we relax the parametric assumptions. We consider a setting with triangular system models and do not impose any functional restrictions on the outcome equation, while we still assume that the instrumental variables are valid. Under these assumptions, we are able to point identify our parameters of interest. 

\subsection{Relaxing Parametric Assumptions} \label{sec:relax_parametric_assumption}
Suppose we observe a random outcome variable $Y$, and a continuous random variable $X$ on $\mathbb{R}$, which is a potentially endogenous treatment variable. We also observe a random vector $Z$ on $\mathbb{R}^k$, $k \geq 1$, which is considered as an instrumental variable. 

Instead of parameterizing the equation of $Y$ on $X$, we consider the following model
\begin{equation} \label{model:tri_model}
    \left\{\begin{aligned}
        Y & =g(X, U), \\
        X & =h(Z, V),
    \end{aligned}\right.
\end{equation}
where $(U, V)$ represent the unobserved heterogeneity. The first equation in Equation \eqref{model:tri_model} is an outcome equation. We do not impose any functional form assumptions on function $g(\cdot, \cdot)$, and the unobservable $U$ can be multi-dimensional. We use $Y_x$ to denote the potential outcome when $X = x$, i.e. $Y_x = g(x, U)$. We use calligraphic characters to denote the supports of random variables. For example, $\mathcal{X}$ indicates the support of $X$.
The second equation is a treatment selection equation for $X$, which takes as arguments the instrument $Z$ and a scalar unobservable $V$. The function $h$ may have a structural interpretation, with potential treatments $X_z = h(z, V)$. However, we can also consider $h$ as a reduced form representation of the conditional distribution of $X$, i.e. define $V=F_{X|Z}(X)$ to be the rank of a unit with treatment value $X$ in the conditional distribution among units sharing their value of $Z$. Then the equation $X=h(Z,V)$ holds with probability one if we define $h$ to yield the conditional quantile function of $X$ given the instrument: $h(z,v)=Q_{X|Z=z}(v)$.\footnote{We follow the standard definition of the quantile funciton, where for any random variable $A$ and $u \in [0,1]$ we define $Q_{A}(u) = \inf \{a: F_{A}(a) \ge u\}$.}

Given $p \in \mathcal{V}$, and $x, x' \in \mathcal{X}$, we define the marginal treatment effect from $x'$ to $x$ at a given value $V=p$ as $\operatorname{MTE}\left(p ; x, x^{\prime}\right) \equiv \mathbb{E}\left[Y_x-Y_{x^{\prime}} \mid V=p\right]$.\footnote{This parameter is defined similarly to the marginal treatment effect with binary treatment, for example, see \cite{heckman2005structural}.} Also, given $x \in \mathcal{X}$, we define the average structural function as $\operatorname{ASF}(x) \equiv \mathbb{E}[Y_x]$. We can use these two parameters to identify other common treatment effects, such as average treatment effects or policy-relevant treatment effects. To identify the parameters of interest, we impose Assumptions \ref{ass:exogeneity} - \ref{ass:Vcdf}.

\begin{assumption}[Independence] \label{ass:exogeneity}
    $Z \indep (U,V)$.
\end{assumption}

\begin{assumption}[Strict monotonicity in V] \label{ass:mon}
    For any $z \in \mathcal{Z}$, where $\mathcal{Z}$ is the support of $Z$, the function $h(z, v)$ is continuous and strictly monotonic in $v$. 
\end{assumption}

\begin{assumption} \label{ass:Vcdf}
    The CDF of $V$ is absolutely continuous and strictly increasing.
\end{assumption}

Assumption \ref{ass:exogeneity} requires that the instrument variable $Z$ is independent of the unobservables $U$ and $V$. Assumption \ref{ass:mon} implies that $X$ is one-to-one mapped to $V$ given $z \in \mathcal{Z}$. Thus, given $Z$, we can invert the function $X = h(Z, V)$ with respect to $V$ and have $V = h^{-1}_{Z}(X)$ almost surely. Assumption \ref{ass:mon} is common in the existing literature. For example, when researchers assume that function $h(Z, V)$ is additively separable in $V$, say $X = f(Z) + V$, Assumption \ref{ass:mon} then holds straightforwardly if $f$ is continuous. Assumption \ref{ass:Vcdf} assures that we can normalize $V \sim \mathcal{U}[0, 1]$, where $\mathcal{U}$ denotes the uniform distribution.\footnote{Note that when $h$ is simply defined to be the conditional quantile function of $h(z,v)=Q_{X|Z=z}(v)$, Assumption \ref{ass:mon} amounts to supposing that the conditional distribution of $X$ has no mass points and has a positive density everywhere on its support. This guarantees that $F_{X|Z}(x)$ is the inverse function of $Q_{X|Z}(v)$ and that $V|Z \sim V \sim \mathcal{U}[0, 1]$. Note that this implies $Z \indep V$ by definition; however Assumption \ref{ass:exogeneity} remains a substantive assumption because it takes $U$ and $V$ to be jointly independent of $Z$.}

For any $x \in \mathcal{X}$ and $z \in \mathcal{Z}$, similar to the setting with a binary treatment, we define propensity score function as 
\begin{equation} \label{eq:p_score}
    P(z, x) \equiv \mathbb{P}(X \leq x \mid Z = z).
\end{equation}

Under Assumptions \ref{ass:exogeneity} - \ref{ass:Vcdf}, we show that the error term $V$ can be identified as the propensity score function on its support.

\begin{lemma}[Identification of $h^{-1}_{Z}$] \label{lemma:pscore}
    Consider the model defined in Equation \eqref{model:tri_model}. If Assumptions \ref{ass:exogeneity}-\ref{ass:Vcdf} hold for $x \in \mathcal{X}$ and $z \in \mathcal{Z}$, the following equality holds:
    \begin{equation} \label{eq:id_pscore}
            P(z, x) 
            = h^{-1}_{z}(x).
    \end{equation}
\end{lemma}
\begin{proof}
    See Appendix \ref{app:proof_lemma3}.
\end{proof}
\noindent Note that as a consequence of \eqref{eq:id_pscore} and Assumptions \ref{ass:mon} and \ref{ass:Vcdf}, we have that
\begin{equation} \label{eq:testableimplication}
P(z, x) \textrm{ is continuous and strictly increasing in } x \textrm{ for all } z \in \mathcal{Z}
\end{equation}
We can identify the function $P(Z, X)$ from the population since $X$ and $Z$ are observed random variables, so that $h_{Z}^{-1}(X)$ is identifiable. Also, Assumption \ref{ass:mon} implies that $V = h^{-1}_{Z}(X)$ almost surely. Therefore, an important implication of Lemma \ref{lemma:pscore} is that we can identify the unobservable $V$ as $V=P(Z, X)$ on the support of $P(Z, X)$. We let $\mathcal{P}$ denote the support of propensity score function $P(Z, X)$ and $\mathcal{P}_x$ denote the support of $P(Z, X)$ given $X = x$ hereafter.

The next lemma proves that $P(Z, X)$ can serve as a control function.

\begin{lemma}[Control function] \label{lemma:identification}
    Consider the model defined in Equation \eqref{model:tri_model}. Under Assumptions~\ref{ass:exogeneity}-\ref{ass:Vcdf},  we can identify the conditional probability of potential outcome $\mathbb{P}(Y_x \in A \mid V = p)$ as 
    \begin{equation}\label{eq:dist}
        \mathbb{P}(Y \in A \mid X = x, P(Z, X) = p) = \mathbb{P}(Y_x \in A \mid V = p)
    \end{equation}
    for $x \in \mathcal{X}$, $p \in \mathcal{P}$, and any set $A \in \mathcal{F}_Y$, where $\mathcal{F}_Y$ denotes the Borel $\sigma$-field generated by $Y$.

    Also, we can identify the conditional expectation of $Y_x$ given $V=p$ as 
    \begin{equation} \label{eq:expectation}
        \mathbb{E}\left[Y \mid X = x, P(Z, X) = p\right] = \mathbb{E} \left[Y_x \mid V = p\right]
    \end{equation}
    for $x \in \mathcal{X}$ and $p \in \mathcal{P}$.
\end{lemma}

\begin{proof}
    See Appendix \ref{app:proof_lemma4}.
\end{proof}

 \noindent \begin{remark} Lemma \ref{lemma:pscore} and an analog of Lemma \ref{lemma:identification} have previously been established by \cite{imbens2009identification} (henceforth IN). IN consider an augmented setup where $X$ can be a vector $X=(X_1,Z_1)$, where $Z_1$ represent exogenous variables that can enter the outcome equation directly. The second equation of \eqref{model:tri_model} is then replaced by $X_1=h(Z_1,Z_2,V)$, i.e. $Z_2$ represent excluded instruments that do not enter the outcome equation. 
\end{remark}
%Note that IN make a slightly stronger independence assumption. Consider the case of no $Z_1$ for simplicity. {\color{red} Then IN assume that $Z \indep (U,V)$ which implies our Assumption \ref{ass:exogeneity} that $Z \indep U |Z$ and $Z \indep V$, while Assumption \ref{ass:exogeneity} does not imply that $(U,V)$ is necessarily jointly independent of $Z$.}\\

\noindent \begin{remark} A generalization that nests both the model described above and that of IN is $Y=g(X_1,Z_1,U)$ and $X_1=h(Z_1,Z_2,V)$ (as IN do) but where Assumptions \ref{ass:exogeneity} and \ref{ass:Vcdf} all hold conditional on $Z_1$. That is: $(Z_2 \indep U)|Z_1,V$ and $(Z_2 \indep V) |Z_1$, while the CDF of $V$ conditional on $Z_1$ is absolutely continuous and strictly increasing, with probability one (for all $Z_1$). This allows the observed variables $Z_1$ to serve as genuine control variables, rather than as ``included'' exogenous regressors as in IN. In this case all of our identification results carry through conditional on $Z_1$, but we omit this generalization for brevity.
\end{remark}

We now show that in the case with no control variables $Z_1$, there are no testable implications of Assumptions \ref{ass:exogeneity}-\ref{ass:Vcdf}, beyond some regularity and support conditions that can be directly verified. This stands in contrast to parametric models like those studied in Section \ref{sec: analytical_framework}, which we have shown have testable implications that can be rejected by the data.

The following condition will be useful in establishing Proposition \ref{prop:nofurther}:
\begin{condition} \label{ass:onetoone}
For any fixed $v \in \mathcal{P}$, let $h^*_v(z) = Q_{X|Z=z}(v)$:
\begin{enumerate}
    \item $h^*_v$ is surjective ("onto"). That is, for any $x \in \mathcal{X}$, there exists $ z \in \mathcal{Z}$ such that $h^*_v(z)=x$.
    \item $h^*_v$ is injective ("one-to-one"). That is, for any $x \in \mathcal{X}$, there exists at most one $z \in \mathcal{Z}$ such that $h^*_v(z)=x$.
\end{enumerate}
\end{condition}
\noindent Given Assumption \ref{ass:Vcdf}, part one of Condition \ref{ass:onetoone} is equivalent to what IN call ``common support''.\footnote{To see this, note that the support of $V$ is $\mathcal{P}$, and the common support condition of IN says that $supp\{P(x,Z)\}=\mathcal{P}$ for all $x \in \mathcal{X}$, in other words for all $v \in \mathcal{P}$ and $x \in \mathcal{X}$, there exists a $z \in \mathcal{Z}$ such that $P(x,z)=v$. Meanwhile, the first part of Condition \ref{ass:onetoone} says that $v \in \mathcal{P}$ and $x \in \mathcal{X}$, there exists a $z \in \mathcal{Z}$ such that $Q_{X|Z=z}(v)=x$. The two are the same given Assumption \ref{ass:Vcdf}.} IN show that all of the quantiles of $g(x,U)$ are identified if Condition \ref{ass:onetoone}.1 holds along with Assumptions \ref{ass:exogeneity}-\ref{ass:Vcdf}. Condition \ref{ass:onetoone} overall holds naturally in parametric models in which $h^*(x,v)$ is additively separable and strictly increasing in $v$, for example, for example if $h_v(z) = \pi \cdot z + \phi(v)$ for some function $\phi$ and $\pi \ne 0$. More generally, both parts of Condition \ref{ass:onetoone} can be directly tested in the data, since they are features of the observed joint distribution of $X$ and $Z$. 

Together, parts 1 and 2 of Condition  \ref{ass:onetoone} say that the function $h^*_v$ is a bijection, for each $x \in \mathcal{X}$ there exists exactly one $z \in \mathcal{Z}$ such that $h^*_v(z)=x$ and vice-versa. A useful implication of Condition \ref{ass:onetoone} for the result that follows is that it enables us to rewrite events written in terms of $Z$ as events written in terms of $X$, for a fixed value of $P(Z,X)$. Specifically, for any Borel set $A$, the event $Z \in A$ is equivalent to the event $X \in h^*_{P(Z,X)}(A)$ conditional on $P(Z,X)$, where we define $h^*_v(A):=\{h^*_v(A): z \in A\}$.\footnote{To see this note that given Condition \ref{ass:onetoone} and conditional on the value of $P(Z,X)$, we have that $X=h^*_{P(Z,X)}(z)$ if and only if $Z=z$, since $h^*_{P(Z,X)}(Z)=Q_{X|Z}(F_{X|Z}(X))=X$ and for any $z' \ne z$ we have by Condition \ref{ass:onetoone} that $h^*_{P(Z,X)}(z') \ne h^*_{P(Z,X)}(z)$.}

%Condition \ref{ass:onetoone} also implies that the function $h^*_v(A)$ is measurable.

Now we are ready to state the main result on the testability of this relaxed IV model:
\begin{proposition} \label{prop:nofurther}
    Suppose that the data satisfies Condition \ref{ass:onetoone}. Then there are no further testable implications of Assumptions \ref{ass:exogeneity}-\ref{ass:Vcdf} together with the model structure  \eqref{model:tri_model}, beyond Equation \eqref{eq:testableimplication}.
\end{proposition}
\begin{proof}
See Appendix \ref{proof:prop:nofurther}.
\end{proof}

Proposition \ref{prop:nofurther} contrasts with results that hold in settings without functional form restrictions but where the treatment is discrete \citep{kitagawa2015test, Huber2015TestingConstraints, mourifie2017testing}, for example the LATE model. In this model, under monotonicity (no-defiers) and random assignment of the instrument, the joint probability of being a complier and having a potential outcome in a given Borel set is point-identified. Testability comes from the non-negativity of this joint probability.   
This approach relies on the ability to vary the instrument with the treatment value fixed, while maintaining overlap between the first-stage types. 
%that contribute to the moments estimated for different instrument values. 
Given Condition \ref{ass:onetoone}, this is not possible in the model of this section because there are one-to-one mappings between $X$ and $Z$ and between $X$ and $V$. %cannot be independently varied for a fixed individual indexed by first-stage type~$V$.

Proposition \ref{prop:nofurther} can also be seen as an extension of the result of \citet{gunsilius2021nontestability}, who shows that a still weaker nonparametric IV model in which Assumption \ref{ass:mon} is dropped and $V$ is allowed to be multivariate lacks any testable implications whatsoever.

\begin{remark} If Condition \ref{ass:onetoone} does not hold, which may be expected to occur especially when $Z$ is multidimensional, we have testable implications for Assumptions \ref{ass:exogeneity} - \ref{ass:Vcdf}. If those assumptions hold, and there exist $z, z' \in \mathcal{Z}, z \neq z'$ such that $P(z, x) = P(z', x) = p$ given $x$, based on conclusions in Lemma \ref{lemma:identification},
    \begin{equation*}
        \begin{aligned}
        & \mathbb{P}(Y \in A \mid X = x, Z = z) \\
        =& \mathbb{P}(Y \in A \mid X = x, P(z, X)=p) \\
        =& \mathbb{P}(Y_x \in A \mid V = p) \\
        =& \mathbb{P}(Y \in A \mid X = x, P(z', X)=p) \\
        =& \mathbb{P}(Y \in A \mid X = x, Z = z').
        \end{aligned}
    \end{equation*}
    Therefore, a testable implication for Assumptions \ref{ass:exogeneity} - \ref{ass:Vcdf} is that $Y \mid \{X = x, Z = z\}$ has the same distribution as $Y \mid \{X = x, Z = z'\}$ if two different instrument values $z, z'$ give the same propensity score evaluated at a fixed $x$. 
\end{remark}

In our setting, $X$ is potentially endogenous, so that the equality $\mathbb{P}(Y_x \in A) = \mathbb{P}(Y \in A \mid X = x)$ may not hold. Lemma \ref{lemma:identification} proves that $P(Z, X)$ serves as a control function. Once conditioning on $P(Z, X)$, $X$ is independent of the potential outcome $Y_x$. Also, since $\mathbb{P}(Y \in A \mid X = x, P(Z, X) = p)$ can be identified in the population level, $\mathbb{P}(Y_x \in A \mid V = p)$ is identifiable. Similarly, we can identify the conditional mean of potential outcome $\mathbb{E}\left[Y_x \mid V = p\right]$, which enables us to identify two parameters of interest, $\operatorname{MTE}\left(p ; x, x^{\prime}\right)$ and $\operatorname{ASF}\left(x\right)$.

\begin{theorem}\label{thm:MTE} (Identifications of MTE and ASF)
    Consider the model defined in Equation \eqref{model:tri_model}. Suppose that Assumptions \ref{ass:exogeneity}-\ref{ass:Vcdf} hold. Then, we can identify the $\operatorname{MTE} (p; x, x')$ as 
    \begin{equation} \label{point_id:MTE_x}
        \operatorname{MTE} (p; x, x') = \mathbb{E}\left[Y \mid X = x, P(Z, X) = p\right] - \mathbb{E}\left[Y \mid X = x', P(Z, X) = p\right]
    \end{equation}
    for any $x \in \mathcal{X}$ and $p \in \mathcal{P}$.

    If $\mathcal{P}_x = [0, 1]$ for any $x$, we can point identify the average structure function as 
    \begin{equation} \label{point_id:ASF_x}
        \operatorname{ASF} (x) = \int_0^1 \mathbb{E}\left[Y \mid X = x, P(Z, X) = p\right] d p.
    \end{equation}

    If $\mathcal{P}_x$ does not have full support but $\mathcal{P}_x \equiv [\underline{p}_x, \bar{p}_x]$, where $\underline{p}_x \geq 0$ and $\bar{p}_x \leq 1$.\footnote{The results can be generalized to cases where the support of the propensity score is not a single interval under some regularity conditions.} Then we can partially identify the average structure function $\operatorname{ASF} (x)$ as 
    \begin{equation}
        \begin{bmatrix} \label{partial_id:MTE_x}
            \int_{[\underline{p}_x, \bar{p}_x]} \mathbb{E}\left[Y \mid X = x, P(Z, X) = p\right] d p + Y_l (1 - \bar{p}_x + \underline{p}_x), \\  
            \int_{[\underline{p}_x, \bar{p}_x]} \mathbb{E}\left[Y \mid X = x, P(Z, X) = p\right] d p + Y_u (1 - \bar{p}_x + \underline{p}_x)
        \end{bmatrix}.
    \end{equation}
     provided that $Y_x$ has bounded support $[Y_l, Y_u]$ for any $x \in \mathcal{X}$.
\end{theorem}

\begin{proof}
    Theorem \ref{thm:MTE} follows from Lemma \ref{lemma:identification} and definitions of $\operatorname{MTE} (p; x, x')$ and $\operatorname{ASF} (x)$.
    %{\color{red} Could we say anything about the sharpness of the bounds in (29)? I think the bounds are sharp. I'm not sure if a proof is necessary.}\\ %{\color{green}I agree that think they are sharp, but probably there's also not too much benefit to saying so -Len}
    %\noindent {\color{green} 
    %For a given $x$ and $p$, does there always exist a $Z$ such that $x = Q_{X|Z}(p)$? Suppose there is and it is unique, $Z=\phi(x,p)$. Then $\mathbb{E}\left[Y \mid X = x, P(Z, X) = p\right] = \mathbb{E}[Y|X=x,Z=\phi(x,p)]$.
    
   % }
\end{proof}

From Lemma \ref{lemma:identification} and Theorem \ref{thm:MTE}, we can (partially) identify other parameters of interest, such as average treatment effects $\operatorname{ATE}(x, x') \equiv \mathbb{E}[Y_x - Y_x']$ and distributions of potential outcomes $\mathbb{P}\left(Y_x \in A\right)$.

\section{Empirical Illustrations}\label{sec: empirical}

In this section, we apply our testing procedure to check the instrument validity in two empirical studies. The results show that our testing method rejects the validity of instrumental variables in some settings. Therefore, we recommend researchers test the instrument validity using our test when working with parametric separable models.  

\subsection{Testing the Validity of the Bartik Instrument in \cite{card2009immigration}} \label{sec:test_card09}

In this section, we revisit \cite{card2009immigration} and examine the validity of the instrumental variable used in the paper. \cite{card2009immigration} studies the impact of immigration on wage imbalances in the United States. Particularly, we revisit his analysis on substitution between immigrants and natives within the same skill group (Table 6 in \cite{card2009immigration}). \cite{card2009immigration} uses the 1980-2000 census data and the combined 2005 and 2006 American Community Surveys to construct panel data with city-level labor market variables in 1980, 1990, 2000, 2005, and 2006. He focuses on the 124 largest MSAs or PMSAs in the US as of 2000.

To estimate the elasticity of substitution, the paper focuses on a linear regression
\begin{equation}
    r_{M j k}-r_{N j k}= \beta_0 + \beta_1 \log \left[S_{M j k} / S_{N j k}\right] + \beta_2' \mathbf{X}_{j}  + u_{j k},
    \label{model:card2009}
\end{equation}
where $r_{M j k}$ represents the mean wage residual for immigrant men in the city $j$ and skill group $k \in \{\textit{hs}, \textit{coll}\}$ ($k = \textit{hs}$ stands for high school or equivalent level, and $k = \textit{coll}$ for college or equivalent), $r_{N j k}$ represents the mean wage residual for native men in the city $j$ and skill group $k$, $S_{M j k} / S_{N j k}$ denotes the ratio of immigrant to native working hours in the city $j$ and skill group $k$, including both men and women workers, and $\mathbf{X}_{j}$ is a vector of city-level controls including log city size, college share, mfg. share in 1980 and 1990, and mean wage residuals for all natives and all immigrants in 1980. By the model structure in Equation \eqref{model:card2009}, the parameter $\beta_1$ measures the negative inverse elasticity of substitution between immigrants and natives in skill group $k$ of city $j$.

However, labor supply may respond to unobserved demand shocks (e.g., new constructions), which could also affect relative earnings. This situation may lead to endogeneity of the variable $S_{M j k} / S_{N j k}$ in this model. To address this potential endogeneity problem, the paper proposes an instrument for the ratio of immigrant to native working hours, $S_{M j k} / S_{N j k}$. The instrument is constructed as
\begin{equation}
    B_{j k} \equiv \frac{\left(\sum_m \lambda_{m j} M_m \delta_{m k}\right)}{P_j},
    \label{iv:card2009}
\end{equation}
where $\lambda_{m j} \equiv N_{m j} / N_m$ is the fraction of previous immigrants from country $m$ that live in city $j$, $N_m$ denotes the number of previous immigrants from country $m$ in the United States, while $N_{m j}$ denotes the number of previous immigrants from country $m$ living in city $j$. $M_m$ stands for the current number of immigrants from country $m$ to the United States, $P_j$ represents the number of the whole population in city $j$, and $\delta_{m k}$ represents the fraction of immigrants from country $m$ that are in skill group $k$. Specifically, to calculate the instrument $B_{j k}$ for immigration labor supply in city $j$ in 2000, \cite{card2009immigration} uses national inflows of immigrants from 38 source countries/country groups over the period from 1990 to 2000 to construct $\lambda_{m j}$, and the shares of each group observed in each city in 1980 to construct $\delta_{m k}$. The intuition of choosing the instrument $B_{j k}$ is that since several studies showed that new immigrants tend to move to the same cities as earlier immigrants, $\lambda_{m j} M_m$ can serve as a prediction for the current immigrants from country $m$ to the city $j$ and could also be independent to the current unobserved demand shocks. Also, the instrument does not correlate with local shocks by such construction. The instrument $B_{j k}$ shares the same idea as the instrument in \cite{bartik1991benefits} and is widely applied across labor economics and international trade.

To identify parameters $(\beta_0, \beta_1, \beta_2')'$, we need to impose the IV assumptions. In this specific case, we need $\mathbb{E} [u_{j k} \mid B_{j k}] = 0$ and $\operatorname{Cov}(S_{M j k} / S_{N j k}, B_{j k}) \neq 0$ to hold. When applying our testing method, the results show that the instrument validity is rejected in the linear structure of Eq.\eqref{model:card2009} at all significance levels for both high-school- and college-equivalent groups. For the group with high school and equivalent levels, the test statistic $\hat{\theta}_{1-\alpha}$ in Equation \eqref{sup_estimator} is equal to $0.081 >0$ for all three significance levels $\alpha = 10\%, 5\%, 1\%$. For the group with college and equivalent level, the test statistic $\hat{\theta}_{1-\alpha}$ is $0.044>0$ for all three levels $\alpha = 10\%, 5\%, 1\%$. A potential reason could be that the unobserved demand shocks for immigrants are persistent over time, which leads to the correlation between the instrument $B_{j k}$ and the current demand shock.
Hence, we have seen an example in which Bartik-type instruments may not satisfy the IV assumptions, and our test method provides a credible criterion for researchers to check validity of their instrument in other settings.

\subsection{Testing the Validity of the Price IV in \cite{nevo2012identification}} \label{sec:test_NR12}

In industrial organization, researchers are often interested in estimating the demand for differentiated products. However, the price variable in the demand function is likely to be endogenous, and researchers commonly use instrument variables to address this issue. In this section, we revisit the application in \cite{nevo2012identification}. They pointed out that the commonly used price instrument variable may not satisfy the exogeneity assumption. We apply our method to check their conjecture.

%To obtain the demand function, researchers usually assume that the individual $i$'s indirect utility for product $j$ in market $k$ is linear in the price and observed characteristics with additively separable unobservables,
%\begin{equation}
%    u_{i j k}=p_{j k} \beta+w_{j k}' \Gamma+\xi_{j k}+\varepsilon_{i j k},
%\end{equation}
%where $p_{j k}$, $w_{j k}$ represent the price and observed characteristics of product $j$ in market $k$, $\xi_{j k}$ denotes the unobserved characteristics of product $j$ in the market $k$, and $\varepsilon_{i j k}$ stands for the idiosyncratic preference of individual $i$ for product $j$ in market $k$. Also, as is customary in this literature we assume that $j$ is in an observed choice set $\mathcal{J}$, and normalize the mean utility for products outside the choice set to be $0$, so that $u_{i 0 k}=\varepsilon_{i 0 k}$. If researchers are willing to assume that the idiosyncratic preference shock $\varepsilon_{i j k}$ is i.i.d. and follows a Type-I extreme value distribution, we can obtain the follwing explicit expression of the market share $s_{j k}$ of product $j$ in market $k$:
%\begin{equation}
%   s_{j k}=\frac{\exp \left(p_{j k} \beta+w_{j k}^{\prime} \Gamma+\xi_{j k}\right)}{1+\sum_{j'=1}^J \exp \left(p_{j' k} \beta+w_{j' k} \Gamma+\xi_{j' k}\right)}.
%\end{equation}
Researchers often assume that the difference in the log market shares is linear in the price and observed characteristics,
\begin{equation}
    \log \left(s_{j k}\right)-\log \left(s_{0 k}\right)=p_{j k} \beta+w_{j k}^{\prime} \Gamma+\xi_{j k},
    \label{iv:demand}
\end{equation}
where $p_{j k}$, $w_{j k}$ represent the price and observed characteristics of product $j$ in market $k$, and $\xi_{j k}$ denotes the unobserved characteristics of product $j$ in the market $k$. The parameters of interest are $(\beta, \Gamma')'$. In practice, the price $p_{j k}$ is potentially correlated with the unobserved demand shock $\xi_{j k}$ (e.g., product quality), which leads to the endogeneity problem. The widely used instruments for $p_{j k}$ are the prices for the same product $j$ in other markets. Ideally, the prices for the same product in markets other than $k$ are related to the price $p_{j k}$ through the common supply shocks (e.g., common production costs), and the demand shocks across different markets are independent. However, the independence of demand shocks across different markets could be violated in some situations. For example, the advertisement within a region could affect the product preference in all markets in that region \citep{nevo2012identification}. Therefore, \cite{nevo2012identification} proposed to use prices in other markets as imperfect instrument variables (IIVs), and they study the ready-to-eat cereal industry at the brand-quarter-MSA (metropolitan statistical area) level. 

We apply our testing procedure to check whether the prices in other markets satisfy the identifying assumptions in the data studied by \cite{nevo2012identification}. When we use the average price in the other city as the instrument variable for price $p_{j k}$ and estimate $(\beta, \Gamma')'$ using the two-stage least squares method, the test results show that the instrument variable validity is rejected with imposed parametric assumptions at all three significance levels, as the supremum statistic $\hat{\theta}_{1-\alpha}$ in Equation \eqref{sup_estimator} is equal to $0.17>0$ for $\alpha = 1\%$, $0.19>0$ for $\alpha=5\%$, and $0.20>0$ for $\alpha = 10\%$. Therefore, we confirm the results in \cite{nevo2012identification} that the average price in other markets is not a valid instrument for the price in the ready-to-eat cereal industry.

\subsection{Applying Relaxed Assumptions: \cite{card2009immigration}}
In Section \ref{sec:test_card09}, our conditional moment inequality test rejects the instrument constructed in \cite{card2009immigration}. If one believes that the instrument $B_{jk}$ is independent of the unobserved demand shocks, then the rejection may caused by the misspecification of the model. In Eq.\eqref{model:card2009}, \cite{card2009immigration} assumes that the gap between mean wage residuals, $(r_{M j k}-r_{N j k})$, is linear in the log of immigrant-to-native ratio $\log \left[S_{M j k} / S_{N j k}\right]$, and the observable $u_{j k}$ is additively separable. However, those specifications could be incorrect if the log of the immigrant-to-native ratio has nonlinear effects on the gap between mean wage residuals or if the unobservable enters into the equation in a more complicated way. In this case, we can consider the model \eqref{model:tri_model} instead to relax those parametric assumptions. \cite{card2009immigration} also includes a vector of controls in his analysis. We assume that the controlled covariates affect the outcome, the treatment, and the log instrument linearly. To exclude the linear covariates' effect, we regress $(r_{M j k}-r_{N j k})$, $\log \left[S_{M j k} / S_{N j k}\right]$, and $\log \left(B_{jk}\right)$ linearly on controlled covariates and take residuals as the outcome $Y_k$, the treatment $X_k$, and the instrument $Z_k$ that we are interested. Table \ref{tab:card_sum_stat} presents a summary of statistics of $Y_k$, $X_k$, and $Z_k$ for $k \in \{hs, coll\}$.

\begin{table}[htbp]
  \centering
  \caption{Summary statistics of outcome, treatment, and instrument excluding covariates' effect in \cite{card2009immigration}}
    \begin{tabular}{lrrrrrrr}
    \toprule
    \toprule
          & \multicolumn{1}{c}{Min.} & \multicolumn{1}{c}{p25} & \multicolumn{1}{c}{Med.} & \multicolumn{1}{c}{p75} & \multicolumn{1}{c}{Max.} & \multicolumn{1}{c}{Mean} & \multicolumn{1}{c}{Std. Dev.} \\
          \bottomrule
    $k$ = high school &       &       &       &       &       &       &  \\
    $Y$     & -0.168 & -0.0504 & -0.0066 & 0.0501 & 0.222 & -1.53E-17 & 0.0734 \\
    $X$     & -1.7731 & -0.7128 & -0.0466 & 0.7517 & 1.6006 & 4.92E-16 & 0.8901 \\
    $\log(Z)$     & -1.2991 & -0.5144 & -0.2437 & 0.5928 & 3.1663 & -5.77E-16 & 0.8506 \\
    $k$ = college &       &       &       &       &       &       &  \\
    $Y$     & -0.2083 & -0.0548 & -0.005 & 0.0473 & 0.3104 & 1.17E-17 & 0.0811 \\
    $X$     & -1.2634 & -0.5488 & -0.0205 & 0.5573 & 1.3352 & 1.50E-16 & 0.6594 \\
    $\log(Z)$     & -1.1203 & -0.3929 & -0.0987 & 0.2862 & 2.147 & -6.95E-16 & 0.5785 \\
    \bottomrule
    \end{tabular}
  \label{tab:card_sum_stat}
\end{table}%
  
Since the instrument $Z_{k}$ is believed to be independent of demand shocks, and the first stage equation in Eq.\eqref{model:tri_model} can be defined as a conditional quantile function, Assumptions \ref{ass:exogeneity} - \ref{ass:Vcdf} hold under this setting. Then, applying Lemma \ref{lemma:identification}, we can point identify $\mathbb{P}\left(Y_{k x} \leq y \mid V_k = p\right)$, the conditional distribution of potential wage gaps given the first-stage unobservable, where $Y_{k x}$ stands for the potential wage gaps with the log of immigrant-to-native ratio set to be $x$ in group $k$. For estimation, we first use the local linear regression to estimate the propensity score in Equation \eqref{eq:p_score}, and then plug in the estimated propensity score in Equation \eqref{eq:dist} and use local linear regression to estimate the conditional distribution of the potential outcome. We select $x$ to be the median and the 75th percentile of the treatment for both groups\footnote{As shown in Table \ref{tab:card_sum_stat}, the median of treatment $X$ is -0.05 for the high school group and -0.02 for the college group, and the 75th percentile of the treatment $X$ is 0.75 for the high school group and 0.55 for the college group.} and set $p$ equal to 0.5 to plot the estimated conditional distributions of potential wage gaps in Figure \ref{fig:est_dist_card09}. For example, the upper left panel of Figure \ref{fig:est_dist_card09} shows the conditional distributions of the potential wage gap if the log of immigrant-to-native ratio would equal the median given the first stage error set at 0.5. Notice that in the upper right panel of Figure \ref{fig:est_dist_card09}, the conditional probabilities are not always between 0 and 1. This is due to estimation error. We leave inference to future work.
\begin{figure}[!htt]
        \centering
        \caption{Estimates of the conditional distributions of potential wage gaps in \cite{card2009immigration}}
        \includegraphics[width = \textwidth]{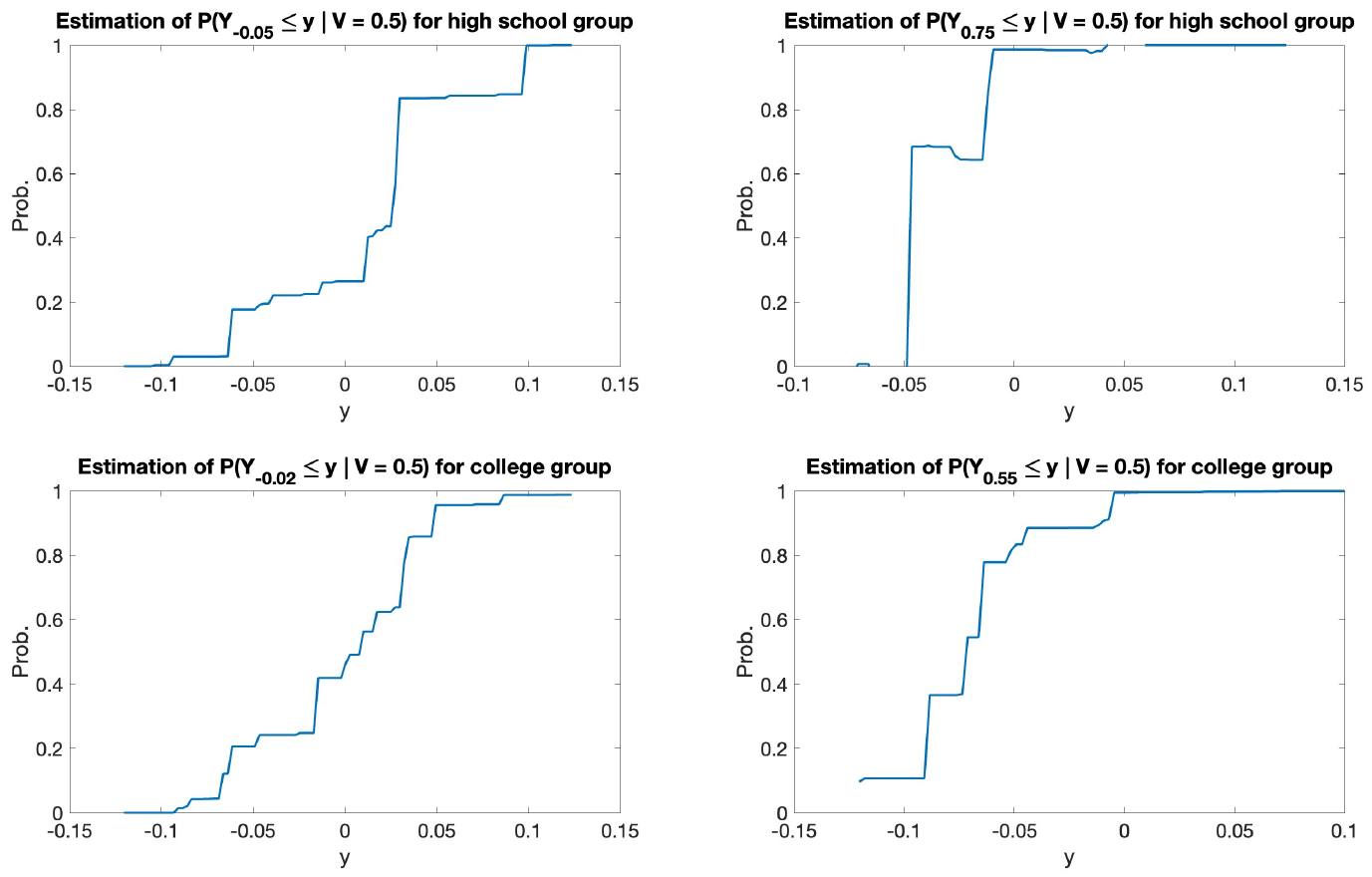}
        \label{fig:est_dist_card09}
\end{figure}

Theorem \ref{thm:MTE} allows us to identify the marginal treatment effect on potential wage gaps if changing the log of the immigrant-to-native ratio from $x'$ to $x$. To estimate the marginal treatment effects, we plug the nonparametrically estimated propensity scores into Equation \eqref{eq:expectation} and use local linear regression to estimate the conditional means in Equation \eqref{eq:expectation}. We study the marginal treatment effects on potential outcome gaps if changing the treatment value from the 25th percentile to the 75th percentile of the sample\footnote{The 25th percentile of the treatment is -0.71 for the high school group and -0.55 for the college group, and the 75th percentile is 0.75 for the high school group and 0.55 for the college group.}, and the marginal treatment effect if changing the treatment value from the 45th percentile to the 55th percentile of the sample for both groups\footnote{The 45th percentile of the treatment is -0.28 for the high school group and -0.17 for the college group, and the 75th percentile is 0.13 for the high school group and 0.09 for the college group.}. Figure \ref{fig:mte_point_card09} presents our estimated marginal treatment effects. For example, the upper left panel of Figure \ref{fig:mte_point_card09} suggests the marginal treatment effect of log immigrant-to-native ratio on the potential wage gap if changing the log immigrant-to-native ratio from its 25th percentile to the 75th percentile, for the first stage error ranging from 0 to 1. We can observe that the marginal treatment effects are heterogeneous in the first stage error $p$.

\begin{figure}[!htt]
        \centering
        \caption{Estimations of marginal treatment effects on potential wage gaps in \cite{card2009immigration}}
        \includegraphics[width = \textwidth]{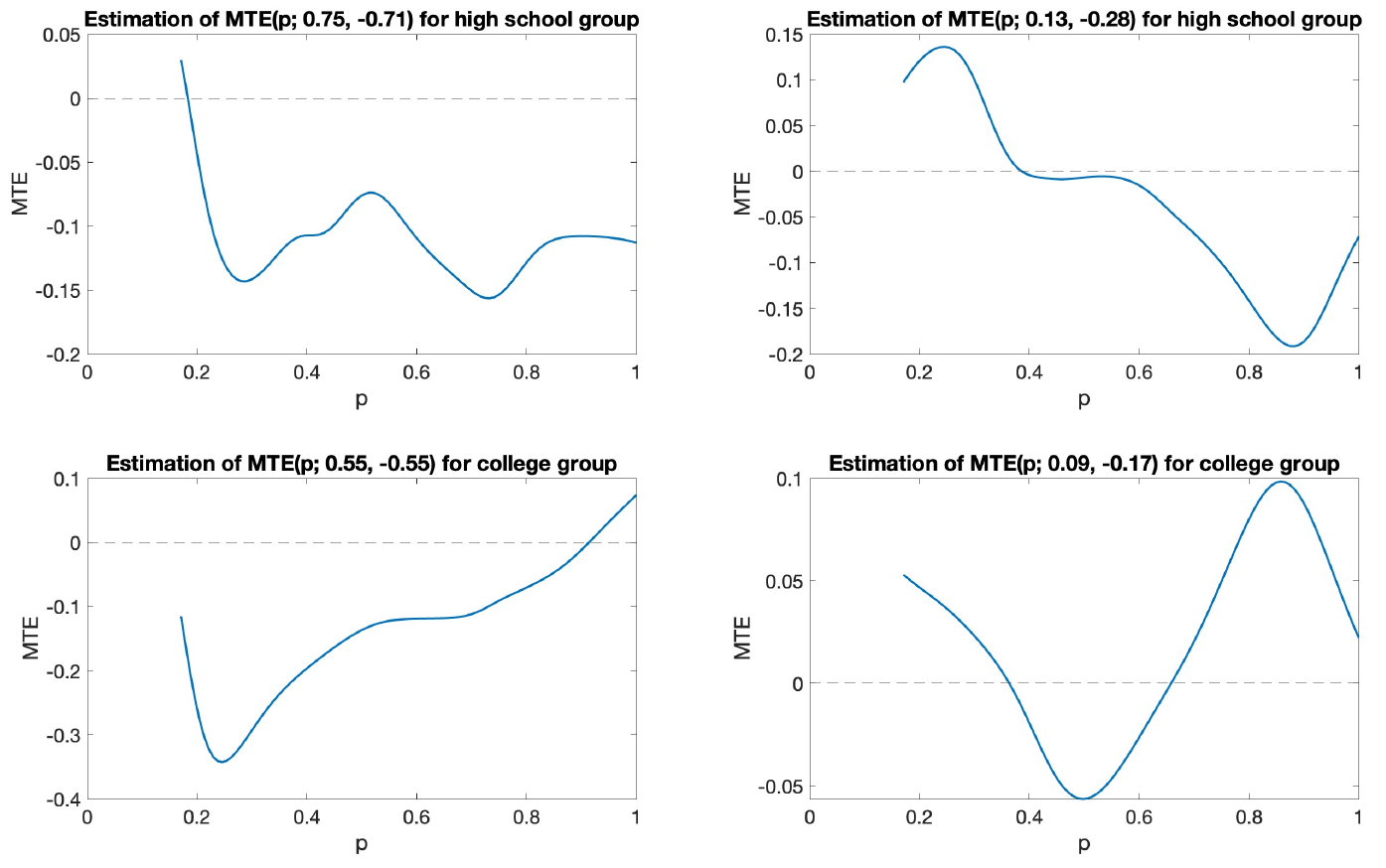}
        \label{fig:mte_point_card09}
\end{figure}

\section{Summary and Discussion}\label{sec: futurework}

We propose an easy-to-implement testing procedure for commonly used identifying assumptions in parametric separable models, including the exogeneity, homoskedasticity, and instrument validity assumptions in linear IV models. We derive a testable implication for these assumptions. We transform this testable implication into a set of conditional moment inequalities and apply the method in \cite{chernozhukov2013intersection} to implement the test. Doing this allows us to use their corresponding Stata packages to perform inference, so that researchers can easily implement our testing procedure in Stata. 

%When estimating the supremum test statistics, \cite{chernozhukov2013intersection} add a precision-correction term to correct for the estimation bias. The such precision-correction term may lead to the conservative of our test. But if the conditional moments are nonflat and smoothed around the maximizer, the estimation of the precision-correction is more accurate. Therefore, compared with previous tests using test statistics in $L^2$ norm, our test is more powerful against alternatives if conditional moments $\mathbb{E}[U \mid X]$ are nonflat and smoothed around maximizers that are larger from zero. However, our test may be less powerful compared with previous tests if conditional moments $\mathbb{E}[U \mid X]$ are flat on the support $\mathcal{X}$ and maximizers are not significantly different from zero. Researchers can take our method as a complement to the previous specification tests. 
We conduct Monte Carlo simulations, and the results suggest that our test performs well in large samples. We also apply our method to test several commonly used instrumental variables in empirical studies. Our proposed test rejects the validity of the Bartik IV used in \cite{card2009immigration}, and the average price in other markets IV questioned in \cite{nevo2012identification}.

Then, we provide solutions to researchers if our test rejects their identifying assumptions. We first relax the parametric assumptions. Instead of imposing structural assumptions on the outcome equation, we consider triangular system models with fully nonparametric outcome equations. We show that if researchers maintain the instrument independence assumption under this setting, they can point identify the distributions of potential outcomes and other standard parameters of interest. 
We further discuss how to ``minimally'' relax the functional form assumptions to the extent where the IV model becomes untestable. Essentially, we find that the IV model is nontestable if no functional form assumption is made on the outcome equation, when there exists a one-to-one mapping between the continuous treatment variable, the instrument, and the first-stage unobserved heterogeneity.

%Next, we further relax the instrumental independence assumption and derive the identified bounds on parameters of interest under our relaxed assumptions. We numerically calculate the identified bounds under our imposed data generating process, and the results verify that our identified bounds can correctly cover the true value of parameters. Also, we use local linear regressions to estimate point- and partially-identified parameters under the imposed data generating process.

%The future work will mainly focus on the following aspects. First, we want to prove whether our identified bounds are sharp. Second, even if we estimate the identified bounds, the estimations cannot be directly used to construct confidence sets of the identified bounds. We need to modify our estimations on the identified bounds and construct confidence sets that cover them uniformly with some prespecified probabilities. Finally, for those empirical works that fail our test of identifying assumptions, we want to implement our identifying strategies and identify their parameters of interest under the relaxed identifying assumptions.

%%%%%%%%%%%%%%%%%%%%%%%%%%%%%%%%%%%%%%%%%%%%%%%%%%%%%%%%%%%%%%%%%%%%%%%%
%%%% Bibliography
%\clearpage
\onehalfspacing
\bibliographystyle{jpe}
\renewcommand\refname{Reference}
\bibliography{Template_Bib}

%%%%%%%%%%%%%%%%%%%%%%%%%%%%%%%%%%%%%%%%%%%%%%%%%%%%%%%%%%%%%%%%%%%%%%%%
%%%% Appendix
%{\noindent\Large\textbf{Appendix}}
\clearpage
\begin{appendices}

\section{Proof of Proposition \ref{prop1}} \label{app:prop1}

\begin{proof}
    To prove the sharpness, we first need to show that if the test does not reject the testable implication, there exist random variables $(\tilde{X}, \tilde{Y}, \tilde{U})$ that can be expressed in the same form as Eq. (\ref{eq:linear_iv}) with some parameters $(\tilde{\beta}_0, \tilde{\beta}_1)$ and satisfy Assumptions \ref{ass:mean_indep} - \ref{ass:relevance} with the instrument variable $Z$, and the vector $(\tilde{X}, \tilde{Y}, Z)$ has the same distribution as $(X, Y, Z)$. 

    We set 
    \begin{equation*}
        \tilde{\beta}_1=\frac{\operatorname{Cov}(Y, Z)}{\operatorname{Cov}(X, Z)}, \quad \tilde{\beta}_0=\mathbb{E}[Y]-\frac{\operatorname{Cov}(Y, Z)}{\operatorname{Cov}(X, Z)} \mathbb{E}[X],
    \end{equation*}
    and construct random variables $(\tilde{X}, \tilde{Y}, \tilde{U})$ as 
    \begin{equation*}
        \tilde{U} = Y - \tilde{\beta}_0 - \tilde{\beta}_1 X, \quad  \tilde{X} = X, \quad \tilde{Y} = \tilde{\beta}_0 + \tilde{\beta}_1 \tilde{X} + \tilde{U}.
    \end{equation*}

    Since $(\tilde{X}, \tilde{Y}, Z) = (X, \tilde{\beta}_0 + \tilde{\beta}_1 X + Y - \tilde{\beta}_0 - \tilde{\beta}_1 X, Z) = (X, Y, Z)$, the random vector $(\tilde{X}, \tilde{Y}, Z)$ has the same distribution as $(X, Y, Z)$. In addition, 
    \begin{equation*}
        \begin{aligned}
            \mathbb{E}[\tilde{U} \mid Z] & = \mathbb{E}[Y - \tilde{\beta}_0 - \tilde{\beta}_1 X \mid Z] \\
            & = \mathbb{E}[(Y - \mathbb{E}[Y]) - \frac{\operatorname{Cov}(Y, Z)}{\operatorname{Cov}(X, Z)}(X - \mathbb{E}[X]) \mid Z] \\
            &= 0
        \end{aligned}
    \end{equation*}
    where the last inequality holds because $(X, Y, Z)$ satisfy the testable implication in Eq.~(\ref{eq:implication}). Furthermore, $\operatorname{Cov}(\tilde{X}, Z) = \operatorname{Cov}(X, Z)$. 
    Therefore, if the testable implication is not rejected, there exists a data generating process $(\tilde{X}, \tilde{Y}, Z)$ that satisfies Eq. (\ref{eq:linear_iv}) and the identifying Assumptions \ref{ass:mean_indep} - \ref{ass:relevance}. Hence, we prove that the testable implication in Eq. (\ref{eq:implication}) is sharp.

    We now want to show that if the test cannot reject the condition in Eq. (\ref{eq:implication}), the actual data generating process may still not satisfy Eq. (\ref{eq:linear_iv}) jointly with Assumptions \ref{ass:mean_indep} - \ref{ass:relevance}. Suppose that condition (\ref{eq:implication}) holds. Define the parameters $(\beta_0^*, \beta_1^*)$ as
    \begin{equation*}
        \beta_1^* = \frac{\operatorname{Cov}(Y, Z)}{\operatorname{Cov}(X, Z)} + 2, \quad \beta_0^* = \mathbb{E}[Y]-\frac{\operatorname{Cov}(Y, Z)}{\operatorname{Cov}(X, Z)} \mathbb{E}[X],
    \end{equation*}
    and the random variables $(X^*, Y^*, U^*)$ as 
    \begin{equation*}
        \begin{aligned}
            & X^* = X, \\
            & U^* = Y - \beta_0^* - \beta_1^* X^*, \\
            %U^* = Y - \left(\mathbb{E}[Y]-\frac{\operatorname{Cov}(Y, Z)}{\operatorname{Cov}(X, Z)} \mathbb{E}[X]\right) - \left(\frac{\operatorname{Cov}(Y, Z)}{\operatorname{Cov}(X, Z)} + 2\right) X^*, \\
            & Y^* = \beta_0^* + \beta_1^* X^* + U^*.
        \end{aligned}
    \end{equation*}
    It is straightforward to check that 
    \begin{equation*}
        (X^*, Y^*, Z) = (X, \beta_0^* + \beta_1^* X + Y - \beta_0^* - \beta_1^* X, Z) = (X, Y, Z).
    \end{equation*}
    So, the random vector $(X^*, Y^*, Z)$ has the same distribution as the observed data $(X, Y, Z)$. However, 
    \begin{equation*}
        \begin{aligned}
            \mathbb{E}[U^* \mid Z] &= \mathbb{E}[(Y - \mathbb{E}[Y]) - \frac{\operatorname{Cov}(Y, Z)}{\operatorname{Cov}(X, Z)}(X - \mathbb{E}[X]) - 2X \mid Z] \\
            &= - 2\mathbb{E}[X \mid Z] \\
            &\neq 0,
        \end{aligned}
    \end{equation*}
    where the second equality follows from Equation (\ref{eq:implication}) and $\mathbb{E}[X \mid Z] \neq 0$ because $\operatorname{Cov}(X, Z) \neq 0$. Therefore, $(X^*, Y^*, Z)$ do not satisfy Assumption \ref{ass:mean_indep}. Hence, the IV model is nonverifiable. 
\end{proof}

\section{The Precision-Corrected Estimator} \label{app:precision-correct}
In Section \ref{sec: analytical_framework}, we are interested in testing
\begin{equation} \label{app:hypothesis}
    H_0: \theta_0 \equiv \sup _{v \in \mathcal{V}} \theta(v) \leq 0 \quad \text { v.s. } \quad H_1: \theta_0>0.
\end{equation}

To estimate $\theta_0$, a naive estimator would simply take the supremum of estimators $\hat{\theta}_j(x)$. However, in this way, we would have a finite upward sample because of estimation error. If we do not correct for this bias, we will tend to over-reject the null hypothesis. To circumvent this problem, \cite{chernozhukov2013intersection} proposed a precision-corrected estimator for $\theta_0$. Using their method to correct the finite sample bias, we can construct the test statistic as 
\begin{equation}
    \hat{\theta}_{1-\alpha} \equiv \sup _{v \in \mathcal{V}}\left\{\hat{\theta}(v)-k_{1-\alpha} \hat{s}(v)\right\},
    \label{app:clrtest}
\end{equation}
where $\hat{s}(v)$ is the standard error of $\hat{\theta}(v)$, $k_{1-\alpha} \hat{s}(v)$ is the critical value, and $k(1 - \alpha)$ is a critical value, which can be considered as a degree of precision correction here. The selection of $k(1 - \alpha)$ is based on the standardized process 
\begin{equation*}
    Z_n(v)=\frac{\theta_n(v)-\hat{\theta}_n(v)}{\sigma_n(v)}.
\end{equation*}
The finite sample distribution of $Z_n$ is unknown, but the paper shows that it can be strongly approximated by a standardized Gaussian process $Z_n^*$ satisfies 
\begin{equation*}
    \bar{a}_n \sup _{v \in \mathcal{V}}\left|Z_n(v)-Z_n^*(v)\right|=o_p(1),
\end{equation*}
where $\bar{a}_n$ is a sequence of constants that $\bar{a}_n \bar{\sigma}_n=O(1)$, $\bar{\sigma}_n:=\sup _{v \in \mathcal{V}} \sigma_n(v)$. The paper shows that $\bar{a}_n$ can be bounded by $\sqrt{\log n}$ if using series or kernel estimators. 

Also, if taking supremum over the entire support $\mathcal{V}$, the inference could be asymptotically valid but conservative. To improve the power of the test, \cite{chernozhukov2013intersection} develop an adaptive inequality selection procedure to select conditional moments close to the boundary. Specifically, conditional moments close to the boundary are those in the set 
\begin{equation*}
    V_n:=\left\{v \in \mathcal{V}: \theta_n(v) \geq \theta_{n 0}-\kappa_n \sigma_n(v)\right\}, 
\end{equation*}
where $\kappa_n:=\kappa_{n, \mathcal{V}}\left(\gamma_n^{\prime}\right)$ with $\gamma_n^{\prime} \nearrow 1$, $\kappa_{n, \mathrm{V}}(\gamma)$ is defined as the $\gamma$-quantile of $\sup _{v \in \mathrm{V}} Z_n^*(v)$, $\kappa_{n, \mathrm{V}}(\gamma):=Q_\gamma\left(\sup _{v \in \mathrm{V}} Z_n^*(v)\right)$. The paper proposes methods that can estimate set $V_n$ with $\hat{V}_n$, then $k_{n, \hat{V}_n}(1-\alpha)$ can be obtained from simulated distribution 
\begin{equation*}
    k_{n, \hat{V}_n}(1-\alpha)=Q_{1-\alpha}\left(\sup _{v \in \hat{V}_n} Z_n^{\star}(v) \mid \mathcal{D}_n\right),
\end{equation*}
where $Z_n^{\star}(v)$ is simulated Gaussian process $Z_n^*$, and $\mathcal{D}_n$ denotes the simulation data.

\section{Validity of Using Plug-in Estimator} \label{app:validity_plug_in}
In this section, we will show that the asymptotic properties of the test are preserved if we use the plug-in estimator for parameter $\beta_0$ and $\beta_1$ in Equation \eqref{eq:implication3}: $\sup_{z}\mathbb E[Y-\beta_0-\beta_1 X \vert Z=z] \leq 0$ and $\sup_{z}\mathbb E[-Y+\beta_0+\beta_1 X \vert Z=z] \leq 0$. The validity of the plug-in approach in other testable implications in this paper can be proved in the same way.

\begin{proof}
    We show the proof of the inequality $\sup_{z}\mathbb E[Y-\beta_0-\beta_1 X \vert Z=z] \leq 0$, as the proof for the other inequality is similar. Define $\mathbb{E}[Y \mid Z = z] = g(z)$, and $r(z)=\mathbb E[X\mid Z=z]$. Then, we have
    $$\sup_{z}\mathbb E[Y-\beta_0-\beta_1 X \vert Z=z] \leq 0 \Longleftrightarrow \sup_{z} \{g(z)-\beta_0-\beta_1 r(z)\} \leq 0$$
    We can estimate $g(z)$ and $r(z)$ using kernel estimation with bandwidth $h_n$ such that $h_n \rightarrow 0$ and $n h_n \rightarrow \infty$ as $n \rightarrow \infty$. We denote by $\hat{g}(z)$ a nonparametric estimator of $g(z)$, and by $\hat{r}(z)$ a nonparametric estimator of $r(z)$. %, and $\hat{g}(z)$ converges in distribution to $g(z)$: 
    %\begin{equation*}
    %   \sqrt{n h_n} \left(\hat{g}(z) - g(z)\right) \xrightarrow{\enskip d\enskip} N(0,V(z)).
    %\end{equation*} 
    We need to show that the asymptotic distribution of $\sqrt{n h_n}\left((\hat{g}(z) -\beta_0-\beta_1 \hat{r}(z))-(g(z)-\beta_0-\beta_1 r(z))\right)$ is unchanged if we replace $\beta_0$ and $\beta_1$ by their estimators $\hat{\beta}_0$ and $\hat{\beta}_1$, respectively. We have
    \begin{eqnarray*}
\sqrt{n h_n}\left((\hat{g}(z) -\beta_0-\beta_1 \hat{r}(z))-(g(z)-\beta_0-\beta_1 r(z))\right) &=&\sqrt{n h_n}\left(\hat{g}(z)-g(z)\right) -\beta_1 \sqrt{n h_n}\left(\hat{r}(z))-r(z))\right),
    \end{eqnarray*}
    and
\begin{eqnarray*}
\sqrt{n h_n}\left((\hat{g}(z) -\hat{\beta}_0-\hat{\beta}_1 \hat{r}(z))-(g(z)-\beta_0-\beta_1 r(z))\right) &=&\sqrt{n h_n}\left(\hat{g}(z)-g(z)\right) -\beta_1 \sqrt{n h_n}\left(\hat{r}(z))-r(z))\right)\\
&-& \sqrt{h_n}\left(\sqrt{n}(\hat{\beta}_0-\beta_0)\right)\\
&-& \sqrt{h_n}\left(\sqrt{n}(\hat{\beta}_1-\beta_1)(\hat{r}(z)-r(z))\right)\\
&-& \sqrt{h_n}\left(\sqrt{n}(\hat{\beta}_1-\beta_1)r(z)\right).
\end{eqnarray*}
Since $\sqrt{n}(\hat{\beta}_0-\beta_0)\xrightarrow{\enskip d\enskip} N(0,V_0)$, $\sqrt{n}(\hat{\beta}_1-\beta_1)\xrightarrow{\enskip d\enskip} N(0,V_1)$, and $\hat{r}(z)\xrightarrow{\enskip p\enskip} r(z),$ we have $\sqrt{n}(\hat{\beta}_0-\beta_0)=O_p(1)$, $\sqrt{n}(\hat{\beta}_1-\beta_1)=O_p(1)$, $\hat{r}(z)-r(z)=o_p(1)$. We also know that $\sqrt{h_n}=o(1)=o_p(1)$. Therefore, by product rule $(O_p(1)*o_p(1)=o_p(1))$, we have
\begin{eqnarray*}
-\sqrt{h_n}\left(\sqrt{n}(\hat{\beta}_0-\beta_0)\right) &=& o_p(1),\\
-\sqrt{h_n}\left(\sqrt{n}(\hat{\beta}_1-\beta_1)(\hat{r}(z)-r(z))\right)&=& o_p(1),\\
-\sqrt{h_n}\left(\sqrt{n}(\hat{\beta}_1-\beta_1)r(z)\right) &=& o_p(1).
\end{eqnarray*}
Finally, since $o_p(1)+o_p(1)=o_p(1),$ we conclude that
\begin{eqnarray*}
\sqrt{n h_n}\left((\hat{g}(z) -\hat{\beta}_0-\hat{\beta}_1 \hat{r}(z))-(g(z)-\beta_0-\beta_1 r(z))\right) &=&\sqrt{n h_n}\left(\hat{g}(z)-g(z)\right) -\beta_1 \sqrt{n h_n}\left(\hat{r}(z))-r(z))\right)\\
&+& o_p(1).
\end{eqnarray*}
Hence, 
\begin{eqnarray*}
\sqrt{n h_n}\left((\hat{g}(z) -\hat{\beta}_0-\hat{\beta}_1 \hat{r}(z))-(g(z)-\beta_0-\beta_1 r(z))\right) &=&\sqrt{n h_n}\left((\hat{g}(z) -\beta_0-\beta_1 \hat{r}(z))-(g(z)-\beta_0-\beta_1 r(z))\right)\\
&+& o_p(1).
\end{eqnarray*}
From the asymptotic equivalence lemma (AEL), $\sqrt{n h_n}\left((\hat{g}(z) -\hat{\beta}_0-\hat{\beta}_1 \hat{r}(z))-(g(z)-\beta_0-\beta_1 r(z))\right)$ and $\sqrt{n h_n}\left((\hat{g}(z) -\beta_0-\beta_1 \hat{r}(z))-(g(z)-\beta_0-\beta_1 r(z))\right)$ have the same asymptotic distribution.

    The above conclusion also holds if we estimate the condition expectations with series estimation. \cite{chernozhukov2013intersection} prove that the convergent rate of series estimator $\hat{g}(z)$ is $\sqrt{\zeta_n^2 / n}$, where $\zeta_n \propto \sqrt{K_n}$ with the number of series terms $K_n$. The similar derivations follow if we relace $\sqrt{h_n}$ with $1 / \zeta_n \rightarrow 0$.
\end{proof}

\section{Simulation Results}\label{app:simulation}

In this section, we present all the simulation results that are mentioned in Section \ref{sec: MonteCarlo}. 

\subsection{Size of Tests} \label{app:size}

To investigate the performance of our testing procedure in the linear model under the exogeneity assumption of a regressor, we consider the specification
\begin{equation}
    Y_i = \beta_0 + \beta_1 X_i + U_i,
    \label{simu:sizeolshomo}   
\end{equation}
where $\beta_0 = 0$ and $\beta_1 = 2$. We randomly draw $X_i \stackrel{\text { i.i.d. }}{\sim} \mathcal{U}[-3, 3]$, $U_i \stackrel{\text { i.i.d. }}{\sim} \mathcal{N}(0,1)$ for $i = 1, \ldots, n$, and $X_i$ is independent of $U_i$, where we use $\mathcal{U}$ to denote the uniform distribution, and $\mathcal{N}$ to denote the normal distribution. 
Table \ref{tab:sizeols} shows the test results for Equation \eqref{test:ols3} in which $Z$ is replaced by $X$.
\begin{table}[!htbp] 
    \caption{Rejection Rates when OLS identifying assumptions hold}
    \centering
    \begin{tabular}{l|ccc}
    & \multicolumn{3}{c}{Significance level}   \\
    Sample size &$10\%$&$5\%$  &$1\%$\\  
    \midrule
    \midrule
    $n=200$ &$10.8\%$&$8.4\%$&$3.6\%$  \\ 
    $n=500$ &$9.0\%$&$5.2\%$&$2.0\%$   \\
    $n=1000$ &$5.8\%$&$1.6\%$&$0.6\%$    \\   
    $n=2000$ &$3.8\%$&$1.6\%$&$0.4\%$    \\ \hline
    \end{tabular}%
    \begin{center}
    \footnotesize{Based on 500 replications. Use series regression to estimate conditional expectations.}
    \end{center}
    \label{tab:sizeols}
\end{table}
As we can see, the test asymptotically controls the nominal size since all three rejection probabilities are respectively less than all three significance levels when the sample size is sufficiently large ($n \geq 1000$).  

Under this DGP, the homoskedasticity assumption holds. We implement our testing procedure on the implications in Equation \eqref{test:olshomo} in which we replace $Z$ by $X$.
Table \ref{tab:sizeolshomo} shows the rejection rates at each significance level with different sample sizes $n$.
\begin{table}[!htbp] 
    \centering
    \caption{Rejection rates when homoskedasticity holds}
    \begin{tabular}{l|ccc}
    & \multicolumn{3}{c}{Significance level}   \\ 
    Sample size &$10\%$&$5\%$  &$1\%$\\  
    \midrule
    \midrule
    $n=200$ &$39.8\%$&$31.6\%$&$22.0\%$  \\ 
    $n=500$ &$19.4\%$&$14.6\%$&$8.8\%$   \\
    $n=1000$ &$18.0\%$&$11.6\%$&$4.8\%$    \\   
    $n=2000$ &$12.0\%$&$6.2\%$&$2.2\%$    \\
    $n=3000$ &$8.6\%$&$6.2\%$&$2.8\%$  \\ \hline
    \end{tabular}%
    \begin{center}
    \footnotesize{Based on 500 replications. Use series regression to estimate the conditional expectations.}
    \end{center}
    \label{tab:sizeolshomo}
\end{table}
When incorporating the homoskedasticity condition into the testing procedure, the nominal size is still controlled when the sample size is large $(n>3000)$. This result is consistent with what the theory suggests.     

% {\color{red}
% ***** I don't understand this example below. It looks like the same as above. Did we test Equation \eqref{test:ols3} instead of \eqref{test:olshomo} while replacing $Z$ by $X$? *****

% To test the exogeneity assumption $(\mathbb E[U \vert Z]=0)$ in the linear model, we consider the following DGP:
% \begin{equation}
%     Y_i = \beta_0 + \beta_1 X_i + U_i,
%     \label{simu:sizeols}   
% \end{equation}
% where $\beta_0 = 0$ and $\beta_1 = 2$. We have conducted 500 replications. For each replication, we set the sample size to be $n$. We randomly draw $X_i \stackrel{\text { i.i.d. }}{\sim} \mathcal{U}[-3, 3]$, $U_i \stackrel{\text { i.i.d. }}{\sim} \mathcal{N}(0,1)$ for $i = 1, \ldots, n$, and $X_i$ is independent of $U_i$, where we use $\mathcal{U}$ to denote the uniform distribution, and $\mathcal{N}$ to denote the normal distribution. Table \ref{tab:sizeols} shows the results.

% \begin{table}[!htbp] 
%     \caption{Rejection rates when OLS identifying assumptions hold}
%     \centering
%     \begin{tabular}{l|ccc}
%     & \multicolumn{3}{c}{Significance level}   \\
%     Sample size &$10\%$&$5\%$  &$1\%$\\  
%     \midrule
%     \midrule
%     $n=200$ &$10.8\%$&$8.4\%$&$3.6\%$  \\ 
%     $n=500$ &$9.0\%$&$5.2\%$&$2.0\%$   \\
%     $n=1000$ &$5.8\%$&$1.6\%$&$0.6\%$    \\   
%     $n=2000$ &$3.8\%$&$1.6\%$&$0.4\%$    \\ \hline
%     \end{tabular}%
%     \begin{center}
%     \footnotesize{Based on 500 replications. Use series regression to estimate conditional expectations.}
%     \end{center}
%     \label{tab:sizeols}
% \end{table}
% }

We also consider a nonlinear least squares specification where the exogeneity assumption of a regressor holds. We generate the data as
\begin{equation}
    Y_i = \beta_0 + \beta_1 X_i^{(\lambda)}+U_i,  
    \label{simu:sizenlls}
\end{equation}
where $\beta_0 = 0$, $\beta_1 = 2$, and 
\begin{equation}
        X_i^{(\lambda)} = \begin{cases}\frac{X_i^\lambda-1}{\lambda}, & \text { if } \lambda \neq 0 \\ \log (X_i), & \text { if } \lambda=0\end{cases}
\end{equation}
for $\lambda = 0, -1, 1$. We draw $U_i \stackrel{\text { i.i.d. }}{\sim} \mathcal{N}(0, 1)$, $X_i \stackrel{\text { i.i.d. }}{\sim} \mathcal{U} (0, 10]$, and $X_i$ is independent of $U_i$. Simulation results are presented in Table \ref{tab:sizenls}.

\begin{table}[!htbp] 
    \caption{Rejection rates when NLLS identifying assumptions hold}
    \centering
      \begin{tabular}{r|l|rrr}
      \multicolumn{1}{c|}{$\lambda$} & \multicolumn{1}{c|}{Sample size} & \multicolumn{3}{c}{Significance level} \\ 
            &       & \multicolumn{1}{c}{10\%} & \multicolumn{1}{c}{5\%} & \multicolumn{1}{c}{1\%} \\ 
      \midrule
      \midrule
      0     & n = 200 & 14.6\% & 11.4\% & 6.6\% \\
            & n = 500 & 8.2\% & 4.8\% & 1.6\% \\
            & n = 1000 & 7.8\% & 4.0\%   & 0.8\% \\
            & n = 2000 & 5.4\% & 2.4\% & 0.4\% \\ 
      \midrule
      -1    & n = 200 & 14.8\% & 11.6\% & 6.0\% \\
            & n = 500 & 9.4\% & 5.6\% & 1.6\% \\
            & n = 1000 & 8.4\% & 4.6\% & 1.0\% \\
            & n = 2000 & 7.8\% & 3.8\% & 0.6\% \\ 
      \midrule
      1     & n = 200 & 16.0\%  & 11.4\% & 8.0\% \\
            & n = 500 & 7.2\% & 4.8\% & 1.0\% \\
            & n = 1000 & 6.4\% & 3.2\% & 0.6\% \\
            & n = 2000 & 4.4\% & 2.0\%   & 0.2\% \\ 
      \bottomrule
      \end{tabular}%
      \begin{center}
      \footnotesize{Based on 500 replications. Use series regression to estimate conditional expectations.}
      \end{center}
      \label{tab:sizenls}
\end{table}%

All the tests under the null hypothesis show that the rejection rates are less than their respective significance levels, which suggests that our testing procedure controls the size.

\subsection{Power of Tests} \label{app:power}
We also apply our testing procedure to test the linear model that violates homoskedasticity. We consider the specification in Equation \eqref{simu:sizeolshomo}, and generate $X_i$, $i = 1, \ldots, 1000$, i.i.d from $\mathcal U[-3,3]$, $U \stackrel{\text { i.i.d. }}{\sim} \mathcal N(0,1+\rho / 9 \cdot X^2)$ with $\rho = 0.1, 0.3, 0.5, 0.7, 0.9$. This DGP does not satisfy the homoskedasticity assumption. Table \ref{tab:powerolshomo} shows the simulation results.
\begin{table}[!htbp] 
    \centering    
    \caption{Rejection rates when homoskedasticity fails}
    \begin{tabular}{l|ccc}
    & \multicolumn{3}{c}{Significance level}   \\ 
    Sample size &$10\%$&$5\%$  &$1\%$\\  
    \midrule
    \midrule
    % $\rho=0.01$ &$18.8\%$&$11.4\%$&$5.6\%$  \\ 
    $\rho=0.1$ &$19.2\%$&$12.6\%$&$5.0\%$   \\
    %$\rho=0.2$ &$32.6\%$&$22.4\%$&$10.6\%$   \\
    $\rho=0.3$ &$53.4\%$&$41.6\%$&$26.4\%$    \\  
    $\rho=0.5$ &$92.8\%$&$87.6\%$&$72.8\%$  \\ 
    $\rho=0.7$ &$99.8\%$&$99.4\%$&$97.2\%$ \\
    $\rho=0.9$ &$100\%$&$100\%$&$100\%$  \\ \hline
    \end{tabular}%
    \begin{center}
    \footnotesize{Based on 500 replications. Use series regression to estimate the conditional expectations.}
    \end{center}
    \label{tab:powerolshomo}
\end{table}

%We consider the DGP of a linear model that violates the exogeneity in Assumption \ref{ass:ols}. 
We consider a specification of a linear model that violates OLS exogeneity assumption. In the model \eqref{simu:sizeolshomo}, we generate $X_i$, $i = 1, \ldots, 1000$, from $X_i \stackrel{\text { i.i.d. }}{\sim} \mathcal U{[-3,3]}$. Then, we generate the error term $U_i = L/\sigma \cdot \phi(X_i / \sigma) + \tilde{U}_i$, where $\tilde{U}_i  = \min\{\max\{-3, V_i\}, 3\}$, and $V_i \stackrel{\text { i.i.d. }}{\sim} \mathcal N(0, 1)$. Similarly, $L$ measures the average deviation of $\mathbb{E}[U_i \mid X_i]$ from zero, and $\sigma$ implies the shape of the conditional mean function. Table \ref{tab:power_ols} presents the rejection rates under this DGP.

\begin{table}[htbp]
    \centering
    \caption{Rejection rate when OLS assumptions fail}
      \begin{tabular}{c|c|rrr}
            &       & \multicolumn{3}{c}{Significance Level} \\
      $L$     & $\sigma$ & \multicolumn{1}{c}{10\%} & \multicolumn{1}{c}{5\%} & \multicolumn{1}{c}{1\%} \\
      \midrule
      \midrule
      \multirow{4}[2]{*}{0.1} & 1     & 16.80\% & 8.60\% & 3.20\% \\
            & 0.5   & 18.80\% & 11.20\% & 3.80\% \\
            & 0.25  & 23.00\% & 15.00\% & 5.40\% \\
            & 0.1   & 24.80\% & 17.40\% & 7.20\% \\
      \midrule
      \multirow{4}[1]{*}{0.5} & 1     & 38.20\% & 29.40\% & 11.60\% \\
            & 0.5   & 85.80\% & 78.00\% & 57.00\% \\
            & 0.25  & 100.00\% & 100.00\% & 99.20\% \\
            & 0.1   & 99.80\% & 99.60\% & 99.40\% \\
      \midrule
      \multirow{4}[0]{*}{1} & 1     & 85.40\% & 77.40\% & 50.80\% \\
            & 0.5   & 100.00\% & 100.00\% & 100.00\% \\
            & 0.25  & 100.00\% & 100.00\% & 100.00\% \\
            & 0.1   & 100.00\% & 100.00\% & 100.00\% \\
        \bottomrule
      \end{tabular}%
      \begin{center}
        \footnotesize{Based on 500 replications. Use series regression to estimate the conditional expectations.}
    \end{center}
    \label{tab:power_ols}%
\end{table}%  

We also consider a nonlinear model where exogeneity assumption fails. We use the same specification in Equation (\ref{simu:sizenlls}) with $\lambda = 0$. We set $X_i = \max\{\tilde{X}_i, 0\}$ with $\tilde{X}_i \stackrel{\text { i.i.d. }}{\sim} \mathcal{U}[-3, 3]$. We also generate $U_i = L / \sigma \cdot \phi(\tilde{X}_i /\sigma) + \tilde{U}_i$, where $\tilde{U}_i = \min\{\max\{-3, V_i\}, 3\}$ with $V_i \sim \mathcal{N}(0, 1)$. Clearly, the exogeneity assumption fails in this DGP. 

Table \ref{tab:power_nls} contains the simulation results in this situation. As discussed in Section 4, our testing procedure becomes more powerful when the deviation from zero is more significant. Also, our test has much higher rejection rates when $\mathbb{E}[U_i \mid X_i]$ is nonflat.

\begin{table}[htbp]
    \centering
    \caption{Rejection rate when NLLS assumptions fail in Box-Cox model}
      \begin{tabular}{c|c|rrr}
            &       & \multicolumn{3}{c}{Significance Level} \\
      $L$     & $\sigma$ & \multicolumn{1}{c}{10\%} & \multicolumn{1}{c}{5\%} & \multicolumn{1}{c}{1\%} \\
      \midrule
      \midrule
      \multirow{4}[2]{*}{0.1} & 1     & 15.20\% & 10.80\% & 4.00\% \\
            & 0.5   & 15.80\% & 10.80\% & 4.00\% \\
            & 0.25  & 16.80\% & 10.60\% & 3.80\% \\
            & 0.1   & 23.60\% & 15.80\% & 8.20\% \\
      \midrule
      \multirow{4}[1]{*}{0.5} & 1     & 20.40\% & 13.40\% & 6.00\% \\
            & 0.5   & 23.20\% & 16.40\% & 7.60\% \\
            & 0.25  & 30.40\% & 21.00\% & 9.20\% \\
            & 0.1   & 98.40\% & 98.00\% & 94.40\% \\
      \midrule
      \multirow{4}[0]{*}{1} & 1     & 33.40\% & 24.60\% & 14.00\% \\
            & 0.5   & 41.20\% & 31.20\% & 18.60\% \\
            & 0.25  & 74.20\% & 63.00\% & 35.20\% \\
            & 0.1   & 100.00\% & 100.00\% & 100.00\% \\
        \bottomrule
      \end{tabular}%
      \begin{center}
        \footnotesize{Based on 500 replications. Use local linear regression to estimate the conditional expectations.}
    \end{center}
    \label{tab:power_nls}%
\end{table}%

\subsection{Comparison With Overidentification Tests} \label{app:comparision}

This section compares our testing method with common overidentification tests when instruments fail the mean independence assumption in linear and nonlinear models. We first consider the same linear model in Section \ref{sec:power}. Table \ref{tab:power_iv_overid} shows the rejection rates of our test and Sargan's test based on $500$ simulation times. Compared with our testing method, under the same DGP, the overidentification test has lower rejection rates in most cases. For example, under the alternative $U_i=L / \sigma \cdot \phi\left(Z_i / \sigma\right)+\tilde{U}_i$ with $L = 0.5$ and $\sigma = 0.1$, our test has $100\%$ rejection rates at all levels $\alpha = 10\%, 5\%, 1\%$, while the overidentification test has rejection rate $77.40\%$ at level $\alpha = 10\%$, $68.80\%$ at level $\alpha = 5\%$, and $44.40\%$ at level $\alpha = 44.40\%$. The results show that the mean independence assumption may not hold even though the overidentification test does not reject it. Therefore, if researchers apply the overidentification test in practice and cannot reject the mean independence assumption, the assumption still may fail, and we suggest using our testing procedure to test the instrument validity.

\begin{table}[htbp]
    \centering
    \caption{Rejection rates when instrument assumptions fail in linear model with overidentification test}
      \begin{tabular}{c|c|rrr}
            &       & \multicolumn{3}{c}{Significance Level} \\
      $L$     & $\sigma$ & \multicolumn{1}{c}{10\%} & \multicolumn{1}{c}{5\%} & \multicolumn{1}{c}{1\%} \\
      \midrule
      \midrule
      \multirow{4}[2]{*}{0.1} & 1     & 13.60\% & 9.60\% & 4.00\% \\
            & 0.5   & 16.20\% & 7.40\% & 1.20\% \\
            & 0.25  & 10.60\% & 6.20\% & 2.00\% \\
            & 0.1   & 12.60\% & 6.80\% & 1.40\% \\
      \midrule
      \multirow{4}[2]{*}{0.5} & 1     & 50.80\% & 35.00\% & 17.20\% \\
            & 0.5   & 77.40\% & 66.80\% & 47.60\% \\
            & 0.25  & 81.20\% &  69.40\% & 43.20\% \\
            & 0.1   & 77.40\% & 68.80\% & 44.40\% \\
      \midrule
      \multirow{4}[2]{*}{1} & 1     & 97.20\% & 95.60\% & 88.00\% \\
            & 0.5   & 99.80\% & 99.60\% & 98.80\% \\
            & 0.25  & 99.60\% & 99.60\% & 98.80\% \\
            & 0.1   & 99.60\% & 99.60\% & 99.00\% \\
      \bottomrule
      \end{tabular}%
    \label{tab:power_iv_overid}%
    \begin{center}
        \footnotesize{Based on 500 replications. Use the Sargan's test statistic.}
    \end{center}
\end{table}%

 Meanwhile, we compare our results with the overidentification test under the same nonlinear model in that the instrument violates the mean independence assumption. We select the instrumental function $h(Z)=\left(Z, Z^2, Z^3\right)^{\prime}$ to convert the conditional moment restriction in Assumption \ref{ass:mean_indep} to the unconditional one and estimate the parameters $\beta_0, \beta_1$ using the two-step generalized method of moments. Then we apply Hansen's overidentification test to check the condition $\mathbb{E}\left[h\left(Z_i\right) U_i\right]=0$. We conduct $500$ simulations with different sample sizes and plot the rejection rates of our conditional moment inequality approach and Hansen's test at 5\% significance level in Figure \ref{fig:power_curve_nonlinear}. We can observe a similar pattern in Figure \ref{fig:power_curve_nonlinear} as in Figure \ref{fig:power_curve_linear}. If the instrument mean independence assumption is violated in the nonlinear model, the rejection rates of both tests increase with the sample size and approach 1 as the sample size becomes large, so that both tests are consistent. Also, our conditional moment inequality test achieves a higher power than Hansen's test for any sample size. Therefore, we suggest using our testing procedure even when the overidentification test cannot reject the mean independence assumption in the nonlinear model. We also compare the rejection rates of conditional moment inequality and Hansen's tests for other values of $L$ and $\sigma$ when the mean independence assumption fails. %We present the results in the Appendix \ref{app:comparision}. 

\begin{figure}[!htt]
        \centering
        \caption{Power curves for testing IV assumptions in nonlinear model}
        \includegraphics[width = 0.8\textwidth]{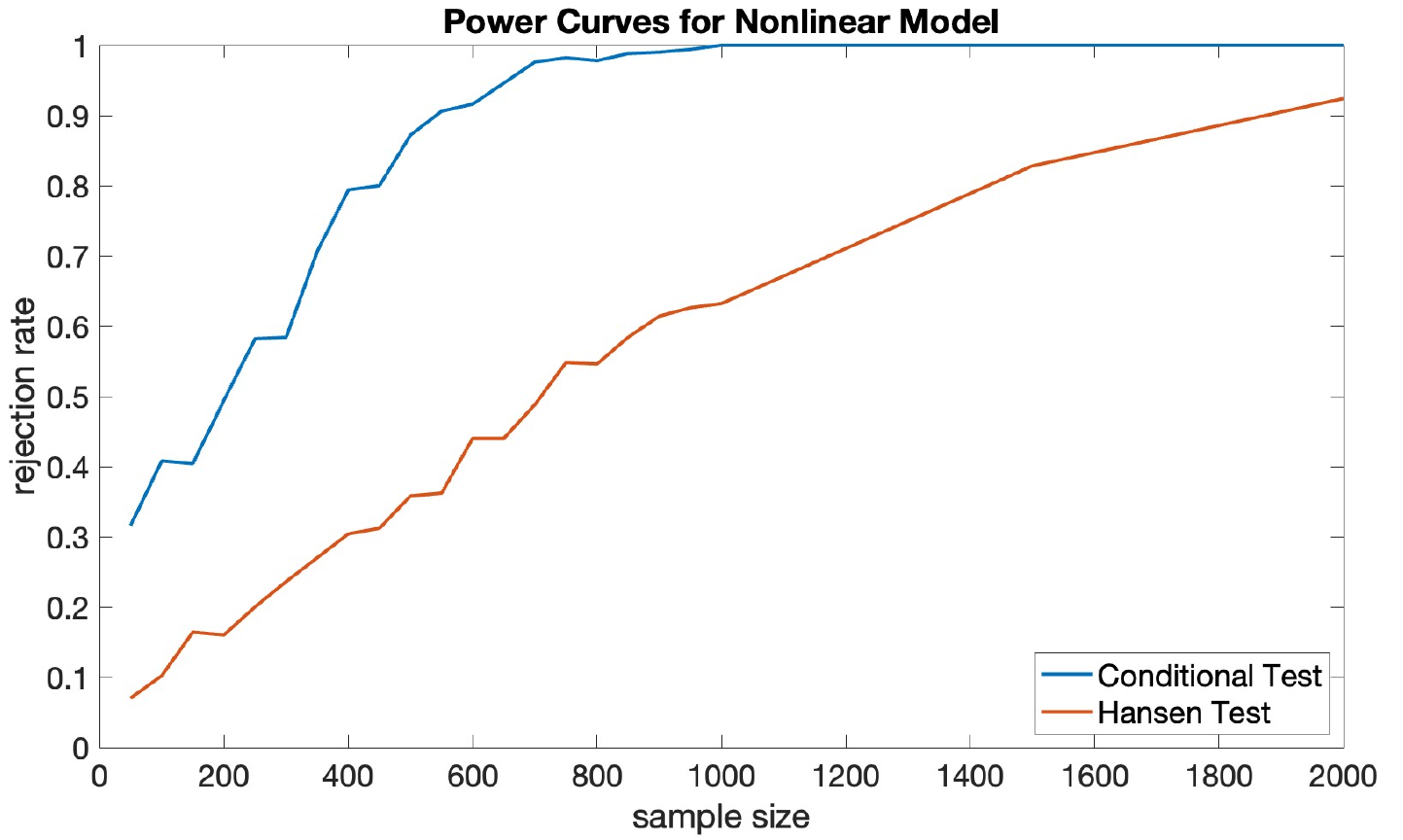}
        \label{fig:power_curve_nonlinear}
\end{figure}

We compare our testing method with Hansen's test under the same alternative in Section \ref{sec:power}. Table \ref{tab:power_iv_overid_boxcox} presents the rejection rates of using Hansen's test based on $500$ simulation times. Comparing with results in Table \ref{tab:power_nls}, the overidentification test has lower rejection rates than our test under most DGPs we consider. Therefore, our test can also achieve higher power in testing the mean independence assumption under nonlinear models than overidentification tests, and we suggest using our testing procedure even when the overidentification test cannot reject the mean independence assumption.

\begin{table}[htbp]
    \centering
    \caption{Rejection rates when instrument assumptions fail in Box-Cox model with overidentification test}
      \begin{tabular}{c|c|rrr}
            &       & \multicolumn{3}{c}{Significance Level} \\
      $L$     & $\sigma$ & \multicolumn{1}{c}{10\%} & \multicolumn{1}{c}{5\%} & \multicolumn{1}{c}{1\%} \\
      \midrule
      \midrule
      \multirow{4}[2]{*}{0.1} & 1     & 11.60\% & 4.40\% & 1.00\% \\
            & 0.5   & 12.60\% & 6.00\% & 1.00\% \\
            & 0.25  & 17.60\% & 8.40\% & 2.40\% \\
            & 0.1   & 15.60\% & 8.20\% & 2.00\% \\
      \midrule
      \multirow{4}[2]{*}{0.5} & 1     & 45.60\% & 32.80\% & 14.80\% \\
            & 0.5   & 72.00\% & 60.60\% & 38.80\% \\
            & 0.25  & 74.20\% &  62.00\% & 37.60\% \\
            & 0.1   & 72.40\% & 60.80\% & 35.20\% \\
      \midrule
      \multirow{4}[2]{*}{1} & 1     & 97.00\% & 91.80\% & 79.60\% \\
            & 0.5   & 99.80\% & 99.60\% & 97.60\% \\
            & 0.25  & 100.00\% & 99.80\% & 98.00\% \\
            & 0.1   & 99.80\% & 99.40\% & 96.60\% \\
      \bottomrule
      \end{tabular}%
    \label{tab:power_iv_overid_boxcox}%
    \begin{center}
        \footnotesize{Based on 500 replications. Use the Hansen's test statistic.}
    \end{center}
\end{table}%

\section{Proof of Identification Results} \label{app:id_results}

 \subsection{Proof of Results in Section \ref{sec:relax_parametric_assumption}} \label{app:proof_6.1}
 
 \subsubsection{Proof of Lemma \ref{lemma:pscore}} \label{app:proof_lemma3}
 For any $x \in \mathcal{X}$ and $z \in \mathcal{Z}$, we have 
     \begin{equation*} 
         \begin{aligned}
             P(z, x) =& \mathbb{P}(h(Z, V) \leq x \mid Z = z) \\
             =& \mathbb{P}(h(z, V) \leq x \mid Z = z) \\
             =& \mathbb{P}(h(z, V) \leq x) \\
             =& \mathbb{P}(V  \leq h^{-1}_{z}(x)) \\
             =& h^{-1}_{z}(x).
         \end{aligned}
     \end{equation*}
 The third line holds under the assumption $Z \indep V$. The fourth line follows from Assumption \ref{ass:mon}, and the last equality is by normalizing $V \sim \mathcal{U}[0, 1]$ under Assumption \ref{ass:Vcdf}.
 
 \subsubsection{Proof of Lemma \ref{lemma:identification}} \label{app:proof_lemma4}
 For any $x \in \mathcal{X}$, $p \in [0,1]$, and set $A \in \mathcal{F}_Y$,
     \begin{equation*} 
         \begin{aligned}
             & \mathbb{P}(Y \in A \mid X = x, P(Z, X) = p) \\
             =& \mathbb{P}(Y_x \in A \mid X = x, P(Z, x) = p) \\
             =& \mathbb{P}(Y_x \in A \mid h(Z, V) = x, h^{-1}_Z(x) = p) \\
             =& \mathbb{P}(Y_x \in A \mid V = h^{-1}_Z(x), h^{-1}_Z(x) = p) \\
             =& \mathbb{P}(Y_x \in A \mid V = p, h^{-1}_Z(x) = p) \\
             =& \mathbb{P}(Y_x \in A \mid V = p),
         \end{aligned}
     \end{equation*}
 where the second equality follows from the conclusion in Lemma \ref{lemma:pscore}, and the last equality holds under $Z \indep U$, which is imposed in Assumption \ref{ass:exogeneity}.
 
 Similarly, we can point identify the conditional expectation $\mathbb{E} \left[Y_x \mid V = p\right]$. For $x \in \mathcal{X}$ and $p \in [0,1]$, we have
     \begin{equation*}
         \begin{aligned}
             &\mathbb{E}\left[Y \mid X = x, P(Z, X) = p\right] \\
             =& \mathbb{E}\left[Y_x \mid X = x, P(Z, x) = p\right] \\
             =& \mathbb{E}\left[Y_x \mid h(Z, V) = x, h^{-1}_Z(x) = p\right] \\
             =& \mathbb{E}\left[Y_x \mid V = h^{-1}_Z(x), h^{-1}_Z(x) = p\right] \\
             =& \mathbb{E}\left[Y_x \mid V = p, h^{-1}_Z(x) = p\right] \\
             =& \mathbb{E}\left[Y_x \mid V = p\right],
         \end{aligned} 
     \end{equation*}
 where the second equality is implied by Lemma \ref{lemma:pscore}, and the last equality holds under $Z \indep U \mid V$.

 \subsubsection{Proof of Proposition \ref{prop:nofurther}}\label{proof:prop:nofurther}
\begin{proof}
    Let $\mathcal{P}_{YXZ}$ denote the joint distribution of $(Y,X,Z)$. To prove the proposition, it is sufficient to show that given a $\mathcal{P}_{YXZ}$ that satisfies \eqref{eq:testableimplication} and Condition \ref{ass:onetoone}, one can always construct variables $U^*$ and $V^*$ and functions $g^*$ and $h^*$ from $\mathcal{P}_{YXZ}$ such that if $U=U^*$,$V=V^*$,$g=g^*$ and $h=h^*$ then Assumptions \ref{ass:exogeneity}-\ref{ass:Vcdf} are satisfied and Equation \eqref{model:tri_model} holds with probability one.

    First, let $V^*:=F_{X|Z}(X)=P(Z,X)$ and $h^*(z,v):=Q_{X|Z=z}(v)$. With these definitions, we have that $X=h^*(Z,V^*)$ a.s.\footnote{See Lemma \ref{lemmaquantile} below for a general proof of this (which requires no assumptions regarding mass points or the existence of a strictly positive density.} Further, since \eqref{eq:testableimplication} implies that $F_{X|Z=z}(X) \sim Unif[0,1]$ we have that $V^*|Z=z \sim Unif[0,1]$ for any $z$, and hence $Z \indep V^*$. Note that Assumptions \ref{ass:mon} and \ref{ass:Vcdf} are satisfied with $V=V^*$ given \eqref{eq:testableimplication}.
    
    Next, define $\theta^* = F_{Y|X,V^*}(Y)$ and let $U^* = (\theta^*,V^*)$. If we denote a generic $u$ by the pair $(p,v)$ and define $g^*(x,u) = Q_{Y|X=x,V^*=v}(p)$, we have that $Y=g^*(X,U^*) = Q_{Y|X,V^*}(F_{Y|X,V^*}(Y))$ with probability one. We have now proved both parts of Equation \eqref{model:tri_model} hold with these definitions.  %{\color{magenta} Thus with probability one $X=h^*(Z,V^*)$. DK: This means: $X=Q_{X|Z}(F_{X|Z}(X))$ a.s. According to the equality $A=Q_A(F_A(A))$ almost surely, we should have $X|Z=Q_{X|Z}(F_{X|Z}(X|Z))$. Do you think $X=Q_{X|Z}(F_{X|Z}(X))$ a.s. also holds?}
    
    It only remains to be seen that $Z \indep U^*|V^*$ so that Assumption \ref{ass:exogeneity} holds with these definitions (since we have already seen that $Z \indep V^*$, this is sufficient for $Z \indep (U^*,V^*)$). Since the $V^*$ component of $U^*=(\theta^*,V^*)$ is degenerate conditional on $V^*$, we only need to check whether $Z \indep \theta^* | V^*$, i.e. for any measurable set $A$, $p \in [0,1]$, and $v \in \mathcal{P}$:
    $$ P(\theta^* \le p|Z \in A,V^*=v) = P(\theta^* \le p|V^*=v).$$  
    Without loss of generality, $\theta^*|X,Z \sim Unif[0,1]$ and hence $\theta^* \indep (X,Z)$.\footnote{This is immediate if $Y$ has no mass points. Otherwise, e.g. if $Y$ is discretely distributed, the random variable $\theta^*=F_{Y|X,V^*}(Y)$ will itself have mass-points at the support points of $Y$. However, one can ``smooth out'' the mass points in $\theta^*$ by re-defining $\theta^*$ to have the distribution $\theta^*|Y=y \sim Unif[\lim_{t \uparrow y} F_{Y|X,V^*}(t),F_{Y|X,V^*}(y)]$. With this definition $\theta^*|X,Z \sim Unif[0,1]$ and it still holds that $Y=g^*(X,U^*) = Q_{Y|X,V^*}(F_{Y|X,V^*}(Y))$ with probability one, as we show in auxillary Lemma \ref{lemmaquantilesmoothed}.   
    %UPDATE for Y: Suppose that $F_{X|Z}(X)$ has a mass point of size $p$ at some value $x$. Let $F = \lim_{t \uparrow x} F_{X|Z}(t)$, so that $F_{X|Z}(x)=F+p$. Note that the CDF of $V^*:=F_{X|Z}(X)$ is not uniformly distributed because it has a mass point of size $p$ at $F+p$. But we can simply smooth out the mass $p$ uniformly over values of $V^*$ between $F$ and $F+p$, defining a new random variable $\tilde{V}^*$. Note that doing so leaves \eqref{model:tri_model} unchanged because $Q_{X|Z}(\tilde{V}^*)=Q_{X|Z}(V^*)$. Doing this for each such mass point (assuming they are isolated, which I think they need to be for a proper probability distribution) results in a random variable $\tilde{V}^*$ which has distribution $Unif[0,1]$ conditional on $Z=z$ for each $z$.
    } This proves useful, as given Condition \ref{ass:onetoone}:
    \begin{align*}
    P(\theta^* \le p|Z\in A,V^*=v)&=P(\theta^* \le p|X\in h^*_v(A),V^*=v) = P(\theta^* \le p|V^*=v)
    %P(\theta^* \le p|Z=z,V^*=v) &= P(F_{Y|X,V^*=v}(Y) \le p|Z=z,V^*=v)\\
    %&=P(F_{Y|X=h^*(z,v),V^*=v}(Y) \le p|Z=z,V^*=v,X=h^*(z,v)]\\
    %&=P(Y \le Q_{Y|X=h^*(z,v),V^*=v}(p)|Z=z,V^*=v,X=h^*(z,v))
    \end{align*}
    We thus have that $Z \indep U^*|V^*$, completing the proof of Proposition \ref{prop:nofurther}. %Thus $P(\theta^* \le p|X=x,V^*=v)$ for any $p,x,v$.
\end{proof}

\subsection{Auxillary supporting results}

\begin{lemma} \label{lemmaquantile}
    Let $X$ be a random variable and $Z$ a random vector defined on a common probability space. Then 
    $$X = Q_{X|Z}(F_{X|Z}(X)) \quad a.s.$$
\end{lemma}
\begin{proof}
    We begin by establishing the simpler claim that
    \begin{equation} \label{eq:unconditionalquantile}
        A = Q_{A}(F_A(A))
    \end{equation}
    almost surely, for any random variable $A$. This result is sometimes used in the literature, but since we cannot find a reference for it, we provide a proof for completeness.  
From Proposition 2 item (2) of \citet{Embrechts_al2013}, $A$ has the same distribution as $\tilde{A}\equiv Q_{A}^{}(U)$ where $U \sim \mathcal{U}_{[0,1]}.$ From the definition of $\tilde{A}$ we can write $$Q_{A}^{}(F_{A}(\tilde{A}))=Q_{A}^{}(F_{A}(Q_{A}^{}(U))).$$
From the definition of $Q_{A}^{}$, we have $F_{A}(Q_{A}^{}(U)) \geq U$. Therefore, since the quantile function is nondecreasing, $Q_{A}^{}(F_{A}(Q_{A}^{}(U))) \geq Q_{A}^{}(U),$ which implies $Q_{A}^{}(F_{A}(\tilde{A})) \geq \tilde{A}$ with probability one.

On the other hand, since $F_{\tilde{A}}(y) \leq F_{\tilde{A}}(y)$, it follows from the definition of the quantile function that for any $y$, $Q_{\tilde{A}}^{}(F_{\tilde{A}}(y)) = \inf\{t: F_{\tilde{A}}(t) \ge F_{\tilde{A}}(y)\} \leq y$. This implies that $Q_{\tilde{A}}^{}(F_{\tilde{A}}(\tilde{A})) \leq \tilde{A}$ with probability one. Finally, since $A$ and $\tilde{A}$ have the same distribution and quantile functions, we then also have that $Q_{A}^{}(F_{A}(\tilde{A})) \leq \tilde{A}.$

As a result, we have both that $Q_{A}^{}(F_{A}(\tilde{A})) \geq \tilde{A}.$ and that $Q_{A}^{}(F_{A}(\tilde{A})) \leq \tilde{A}$, so $Q_{A}^{}(F_{A}(\tilde{A})) = \tilde{A}$ with probability one. It only remains to be seen that this implies \eqref{eq:unconditionalquantile}, that $Q_{A}^{}(F_{A}(A) = A$, with $A$ replacing $\tilde{A}$.

Now, define $S=\left\{y \in \mathbb R: Q_{A}^{}(F_{A}(y)) = y \right\}$. Since $F_{\tilde{A}}=F_{A}$, we have $\mathbb P(A \in S)=\mathbb P(\tilde{A} \in S)$ for any Borel set $S$. Thus
$$\mathbb P(A \in S)=\mathbb P(\tilde{A} \in S)=\mathbb P\left(Q_{A}^{}(F_{A}(\tilde{A})) = \tilde{A}\right)=1.$$ Hence, 
$$1=\mathbb P(A \in S)=\mathbb P\left(Q_{A}^{}(F_{A}(A))=A\right).$$
establishing \eqref{eq:unconditionalquantile}.

To move to the conditional case, now let us consider for any value $z$ the random variable $A$ having distribution function $F_{X|Z=z}$. We have by \eqref{eq:unconditionalquantile} that
    $$\mathbb P\left(A = Q_{X|Z=z}(F_{X|Z=z}(A))\right)=1$$
    which given $A \sim X|Z=z$ is the same as 
    \begin{equation} \label{eq:equalszeroconditionalonz}
        \mathbb P\left(X = Q_{X|Z=z}(F_{X|Z=z}(X))|Z=z\right)=1
    \end{equation}
    By the law of total probability, we have that
    \begin{align*}
    \mathbb{P}\left(X = Q_{X|Z}(F_{X|Z}(X))\right) &= \int \mathbb{P}\left(X = Q_{X|Z}(F_{X|Z}(X))|Z=z\right) \cdot dF_Z(z)\\
    &=\int \mathbb{P}\left(X = Q_{X|Z=z}(F_{X|Z=z}(X))|Z=z\right) \cdot dF_Z(z)
    &=\int 1 \cdot dF_Z(z) = 1
    \end{align*}
    where we've used \eqref{eq:equalszeroconditionalonz} in the last line.
\end{proof}

\begin{lemma} \label{lemmaquantilesmoothed}
    Let $Y$ be a random variable and $X$ a random vector defined on a common probability space. Define random variable $\theta$ to have the distribution $\theta|(Y,X) \sim Unif[\lim_{s \uparrow Y} F_{Y|X}(s), F_{Y|X}(Y)]$. Then 
    $$\theta|X \sim Unif[0,1]$$ and
    $$Y = Q_{Y|X}(\theta) \quad a.s.$$
\end{lemma}
\begin{proof}
    As in Lemma \ref{lemmaquantile}, we begin with the simpler case in which there is no conditioning random vector $X$. To this end, consider a random variable $A$ and define $\theta$ by $\theta|A \sim Unif[\lim_{s \uparrow A} F_{A}(s), F_{A}(A)]$. 

    Let us introduce the shorthand $F_{A}(a^-):=\lim_{s \uparrow a} F_{A}(s)$ for any $a$. The limit $F_{A}(a^-)$ exists for any $a$ since a CDF can only exhibit countable discontinuities.
    
    First observe that for any $a$ and $\tilde{a} \in [F_A(a^-), F_{A}(a)]$ we have that $Q_A(\tilde{a})=Q_A(F_{A}(a)$) since $Q_A(\cdot)$ is constant on the range $[\lim_{s \uparrow a} F_{A}(s), F_{A}(a)]$ (when
    it is not a singleton). Since $\theta \in [F_A(A^-), F_{A}(A)]$ with probability one, we then have that $Q_A(\theta) = Q_A(F_A(A))$ with probability one. By Lemma \ref{lemmaquantile}, $Q_A(F_A(A))=A$ with probability one and combining we thus have $Q_A(\theta)=A$ with probability one. 
    
    Now we turn to showing that $\theta \sim Unif[0,1]$,\footnote{We thank Eric Mbakop for useful conversations on the direction of this part of the proof.} i.e. that
    $$P(\theta \le t) = \begin{cases}
    0 & \textrm { if } t < 0\\
    t & \textrm { if } t \in [0,1]\\
    1 & \textrm { if } t > 1 \end{cases}$$
    By the law of iterated expectations, 
    \begin{equation} \label{eq:cdfthetalie0}
    P(\theta \le t) = \int P(\theta \le t|A=a) dF_A(a) 
    \end{equation}
    and given the definition of $\theta$
    $$P(\theta \le t|A=a) = \begin{cases}
    0 & \textrm { if } t < F_A(a^-) \\
    \frac{t-F_A(a^-)}{F_A(a)-F_A(a^-)} & \textrm { if } t \in [F_A(a^-),F_A(a) ]\\
    1 & \textrm { if } t > F_A(a)  \end{cases}$$

    \noindent Note that this implies $F_A(a^-) \ge 0$ and $F_A(a) \le 1$ for any $a$, the above immediately yields that $P(\theta \le t) = 0$ for $t<0$ and $P(\theta \le t) = 1$ for $t>1$.
    
    Accordingly, let us now focus on any $t \in [0,1]$. We can simplify \eqref{eq:cdfthetalie0} to
    \begin{equation} \label{eq:cdfthetalie}
        P(\theta \le t) = \int_{a:F_A(a^-) \le t} P(\theta \le t|A=a) dF_A(a)  
    \end{equation}
    where the second equality follows because $\theta \ge F_A(a^-)$ with probability one conditional on $A=a$, so $P(\theta \le t|A=a)=0$ for any $a$ such that $F_A(a^-) > t$.
    
    Define $a^*(t) = \sup\{a: F_A(a^-) \le t\}$. Note that we cannot have that $F_A(a^*(t)) < t$, since by right continuity of $F_A(\cdot)$, $F_A(a^*(t)) = \lim_{s \downarrow a^*(t)} F_A(s)$ and thus if $F_A(s) < t$ for some $s > a^*(t)$, it would violate the definition of $a^*(t)$ as a supremum. We are thus left with two cases for a given $t$: either $F_A(a^*(t)) =t$ or $F_A(a^*(t)) > t$.
    
    Consider the first case, in which $F_A(a^*(t)) =t$. Then the condition $a \le a^*(t)$ is equivalent to $F_{A}(a^-) \le t$, so by \eqref{eq:cdfthetalie} then
    $$P(\theta \le t) = \int_{a \le a^*(t)} P(\theta \le t|A=a) dF_A(a) = \int_{a \le a^*(t)} 1 dF_A(a) = P(A \le a^*(t)) = F_A(a^*(t))=t$$
    where the second equality uses that $\theta \le F_A(a)$ with probability one conditional on $A=a$, and $F_A(a^*(t)) =t$ implies that $F_A(a) \le t$ for all $a \le a^*(t)$, so $P(\theta \le t|A=a)=1$ for all $a \le a^*(t)$. %The third equality uses that given $A \le a^*(t)$

    Now let's turn to the second case, in which $F_A(a^*(t)) > t$. This implies that there is a point mass at $a^*(t)$, because if there were no point-mass at $a^*(t)$ we would have $F_A(a^*(t)^-)=F_A(a^*(t))$ and therefore $F_A(a^*(t)^-) > t$, contradicting the definition of $a^*(t)$. Eq. $\eqref{eq:cdfthetalie}$ thus becomes:
    \begin{align*}
    &P(\theta \le t) = P(A=a^*(t))\cdot P(\theta \le t|A=a^*(t))+\int_{a < a^*(t)} P(\theta \le t|A=a) dF_A(a) \\
    &= P(A=a^*(t))\cdot P(\theta \le t|A=a^*(t))+\int_{a < a^*(t)} 1 dF_A(a)\\
    &= P(A=a^*(t))\cdot \frac{t-F_{A}(a^*(t)^-)}{F_A(a^*(t))-F_{A}(a^*(t)^-)}+F_{A}(a^*(t)^-)=t
    \end{align*}
    where the second equality uses that for any $a < a^*(t)$, there must exist an $s>a$ such that $F_A(s^-)\le t$ (by the definition of $a^*(t)$), and thus $F_A(a)\le t$. Therefore $P(\theta \le t|A=a)=1$ for any such $a$. The third equality uses \eqref{eq:cdfthetalie0} and that $P(A=a^*(t)) = F_A(a^*(t))-F_A(a^*(t)^-)>0$.\\
    
    \noindent As with Lemma \ref{lemmaquantile}, we now move to the conditional case by replacing the distribution $F_A$ with a conditional distribution function, in this case $F_{Y|X=x}$ for a given $x$. We now define $\theta$ to have the distribution
    $$\theta|(Y=y,X=x) \sim Unif[\lim_{s \uparrow y} F_{Y|X=x}(s), F_{Y|X=x}(y)]$$
    and the above result then establishes that $\theta|X=x \sim Unif[0,1]$. Thus $\theta|X \sim Unif[0,1]$ with probability one and $\theta \indep X$.

    That $Q_{Y|X}(\theta)=Y$ with probability one follows similarly to the unconditional case. Since $Q_{Y|X=x}(\cdot)$ is constant on the range $[\lim_{s \uparrow y} F_{Y|X=x}(s), F_{Y|X=x}(y)]$, it follows that $Q_{Y|X=x}(\tilde{y})=Q_{Y|X=x}(F_{Y|X=x}(y)$ for any $\tilde{y} \in [\lim_{s \uparrow y} F_{Y|X=x}(s), F_{Y|X=x}(y)]$. We then have that $Q_{Y|X=x}(\theta) = Q_{Y|X=x}(F_{Y|X=x}(Y))$ with probability one, and thus $P(Q_{Y|X}(\theta)=Y|X=x)=1$. Then by the law of total probability $P(Q_{Y|X}(\theta)=Y)=1$, as in Lemma \ref{lemmaquantile}.
    
    \end{proof}

\end{appendices}

\end{document}